\documentclass[draftcls,journal,onecolumn]{IEEEtran}
\usepackage{amssymb,amsmath,amsthm,subfig,cite}
\usepackage{tikz}
\usetikzlibrary{shadows,positioning,shapes.symbols}

\DeclareMathOperator{\col}{col}
\DeclareMathOperator{\row}{row}
\DeclareMathOperator{\diag}{diag}
\DeclareMathOperator{\GL}{GL}
\DeclareMathOperator{\rank}{rank}
\DeclareMathOperator{\shape}{shape}
\DeclareMathOperator{\RCF}{RCF}
\newcommand{\len}{s} 
\newcommand{\mat}[1]{\begin{bmatrix} #1 \end{bmatrix}}
\newcommand{\submodule}[2]{\left[\!\!\left[\genfrac{}{}{0pt}{}{#1}{#2}\right]\!\!\right]_q}
\newcommand{\gauss}[2]{{\genfrac{[}{]}{0pt}{}{#1}{#2}}_q}

\newtheoremstyle{plain}
  {\medskipamount}{\medskipamount}{\rmfamily}{\parindent}{\bfseries}{:}{0.5em}
  {\thmname{#1}\thmnumber{\ #2}\thmnote{\ (#3)}}
\newenvironment{exam}{\myexample}{\qed\endmyexample}
\newtheorem{myexample}{Example}
\makeatletter\@ifundefined{endIEEEproof}{}{\def\qed{\endIEEEproof}}\makeatother
\newtheorem{thm}{Theorem}
\newtheorem{prop}{Proposition}
\newtheorem{lem}{Lemma}
\newtheorem{cor}{Corollary}
\newtheorem{defi}{Definition}

\ifCLASSOPTIONonecolumn
\tikzset{sbox/.style={minimum width=2.2ex,minimum height=2.2ex,inner sep=0pt},
lbox/.style={fill=white,inner sep=1pt}}
\else
\tikzset{sbox/.style={minimum width=2.1ex,minimum height=2.1ex,inner sep=0pt},
lbox/.style={fill=white,inner sep=1pt}}
\fi

\newcommand{\CC}{\mathbb{C}}
\newcommand{\ZZ}{\mathbb{Z}}

\newcommand{\FF}{\mathbb{F}}
\newcommand{\calA}{\mathcal{A}}

\newcommand{\calC}{\mathcal{C}}

\newcommand{\calG}{\mathcal{G}}

\newcommand{\calO}{\mathcal{O}}

\newcommand{\calR}{\mathcal{R}}
\newcommand{\calS}{\mathcal{S}}
\newcommand{\calT}{\mathcal{T}}

\newcommand{\calX}{\mathcal{X}}
\newcommand{\calY}{\mathcal{Y}}

\title{Communication over Finite-Chain-Ring\\Matrix Channels}

\author{\IEEEauthorblockN{Chen~Feng,
Roberto~W.~N\'obrega,
Frank~R.~Kschischang~\IEEEmembership{Fellow,~IEEE},
Danilo~Silva}
\thanks{Manuscript received April 8, 2013; revised February 18, 2014.
This paper was presented in part at the IEEE International
Symposium on Information Theory, Cambridge, MA, July 2012.}%
\thanks{C.~Feng and F.~R.~Kschischang
are with the Dept.~of
Elec. \& Comp. Eng., U. of Toronto, Canada,
\texttt{\{cfeng@eecg,frank@comm\}.utoronto.ca},
R.~W.~N\'obrega and D.~Silva are with the Dept.~of
Elec. Eng, Federal U. of Santa Catarina, Brazil,
\texttt{\{rwnobrega,danilo\}@eel.ufsc.br}.
The work of R.~W.~N\'obrega and D.~Silva was supported in part by
CNPq-Brazil.} }
\IEEEpubid{0000--0000/00\$00.00~\copyright~2014 IEEE}

\begin{document}
\maketitle

\begin{abstract}
Though network coding is traditionally performed over finite
fields, recent work on nested-lattice-based network coding
suggests that, by allowing network coding over certain finite
rings, more efficient physical-layer network coding schemes
can be constructed.  This paper considers the problem of
communication over a finite-ring matrix channel $Y =
AX + BE$, where $X$ is the channel input, $Y$ is the channel
output, $E$ is random error, and $A$ and $B$ are random
transfer matrices.  Tight capacity results are obtained and
simple polynomial-complexity capacity-achieving coding schemes
are provided under the assumption that $A$ is uniform over
all full-rank matrices and $BE$ is uniform over all rank-$t$
matrices, extending the work of Silva, Kschischang and K\"{o}tter
(2010), who handled the case of finite fields.  This extension
is based on several new results, which may be of independent
interest, that generalize concepts and methods from matrices
over finite fields to matrices over finite chain rings.
\end{abstract}


\section{Introduction}

\IEEEPARstart{M}{atrix}  channels provide a useful abstraction
for studying
error control for linear network coding schemes.  Transmitted and
received packets, drawn from some ambient message space $\Omega$, can be
gathered into the rows of a transmitted matrix $X$ and a received matrix
$Y$, respectively, while error packets injected into the network can be
described by the rows of an error matrix $E$.  Due to the nature of
linear network coding, the linear transformation of transmitted packets
$X$ and the linear propagation of error packets $E$ can be modelled as
a multiplicative-additive matrix channel (MAMC), defined via
\begin{equation}
  Y = AX + BE
\label{eqn:MatrixChannel}
\end{equation}
for appropriate transfer matrices $A$, $B$. One typically assumes that
$A$, $B$, and $E$ are random matrices (drawn according to certain
distributions) and independent of $X$.  This type of stochastic model is
appropriate in situations where random network coding is performed and
the error matrix $E$ arises due to decoding errors, rather than from the
malicious actions of an adversary.

When the ambient space $\Omega$ is a vector space over a finite field,
tight capacity bounds and simple, asymptotically capacity-achieving,
coding schemes are developed in \cite{SKK10}, under certain
distributions of $A$, $B$, and $E$.  Similar work along this line can be
found, e.g., in \cite{MU13,Jafari11,Yang10,Nobrega11}.  Prior work on
matrix channels for linear network coding has mainly focused on the
finite-field case.

In this paper, we consider a more general ambient space $\Omega$ of the
form
\begin{equation}
\Omega = T / \langle d_1 \rangle \times T / \langle d_2 \rangle \times
         \cdots \times T / \langle d_m \rangle,
\label{eq:ambient1}
\end{equation}
where $T$ is a sub-ring of $\CC$ forming a principal ideal domain and
$d_1, d_2, \ldots, d_m \in T$ are nonzero non-unit elements. To handle
such an ambient space, we need to generalize the work of \cite{SKK10}
from finite fields to finite chain rings.  The motivation for
considering this generalization arises from nested-lattice
physical-layer network coding
\cite{NG09,NWS07,FSK-submitted,TNBH12,QJ12}, in which the ambient space
$\Omega$ is given precisely in the form of \eqref{eq:ambient1}.  As in
\cite{SKK10}, we gather insight by first studying two variations: the
noise-free multiplicative matrix channel (MMC) $Y=AX$, and the
multiplication-free additive matrix channel (AMC) $Y=X+BE$.

The essential step in handling the MMC over finite fields is based on
the concept of reduced row echelon form (RREF) \cite{SKK10}.  Due to the
presence of zero divisors, the extension to finite chain rings of this
concept is not straightforward.  Whereas over a finite field any echelon
form of a matrix will have the same number of nonzero rows (equal to the
matrix rank), this is not the case for matrices over finite chain rings.
To address this difficulty, several possible extensions of the RREF have
been proposed in the literature, including the Howell form \cite{How86,
Stor00} and the $p$-basis \cite{VSR96}.  In this paper, we use the row
canonical form defined in the dissertation of Kiermaier
\cite{Kiermaier12}, which is itself a variant of the matrix canonical
form described in an exercise in \cite{McDonald74}, and traces back to
earlier ideas of Fuller \cite{F55} and Birkhoff \cite{B35}; see
Section~\ref{sec:canonical} for more details.  This row canonical form is
particularly suitable for studying matrix channels with an ambient space
of the form \eqref{eq:ambient1}.  We provide a new elementary proof for
the existence and uniqueness of this row canonical form.  Based on these
results, we introduce a notion of (combinatorially dominant)
\emph{principal} row canonical forms, which allows us to obtain simple,
capacity-achieving, coding schemes for the MMC.

\IEEEpubidadjcol

The key step in handling the AMC over finite fields is counting the
number of matrices of a given rank $t$. The rank $t$ may be regarded as
a measure of ``noise level'' of the matrix $BE$.  For matrices over
finite chain rings, the concept of ``rank'' is more subtle, and must be
suitably generalized. We first show how the concept of ``shape''---the
appropriate chain-ring-theoretic generalization of dimension---can be
used to indicate the noise level.  We then derive an enumeration result
that counts the number of matrices of a given shape. This enables us to
obtain capacity results and simple capacity-achieving coding schemes for
the MMC.

Building upon the generalizations for the two special cases, we derive
tight capacity bounds and simple, polynomial-complexity, asymptotically
capacity-achieving coding schemes for the MAMC model related to
(\ref{eqn:MatrixChannel}). We also consider several possible extensions
of the MAMC model.

The remainder of this paper is organized as follows.
Section~\ref{sec:motivating} motivates the study
of matrix channels over finite rings.
Section~\ref{sec:preli} reviews some basic facts about finite chain
rings, modules and matrices over finite chain rings.
Section~\ref{sec:canonical} introduces the row canonical form.
Section~\ref{sec:matrix-constraint} presents several enumeration results
and construction methods for matrices over finite chain rings.  These
new results provide us with essential algebraic tools for extending the
work of \cite{SKK10}.    Section~\ref{sec:channel} introduces a
channel-decomposition technique that connects the matrix channels
described in Section~\ref{sec:motivating} to the
algebraic tools developed in Sections~\ref{sec:canonical} and
\ref{sec:matrix-constraint}.  Three basic channel models (MMC, AMC, and
MAMC) are addressed in Sections~\ref{sec:mmc}, \ref{sec:amc} and
\ref{sec:ammc}, respectively, where capacity and coding results are
presented.  Section~\ref{sec:extension} presents possible extensions.
Finally, Section~\ref{sec:conclusion} concludes the paper.

\section{Motivating Examples}\label{sec:motivating}


In this section, we explain how finite rings arise naturally in the
context of nested-lattice-based physical-layer network coding (PNC).  We
then introduce an end-to-end matrix model for wireless relay networks
based on such PNC schemes.

\begin{figure}[h]
\centering
\subfloat[Transmitted constellation]{
\scalebox{0.75}{
\includegraphics{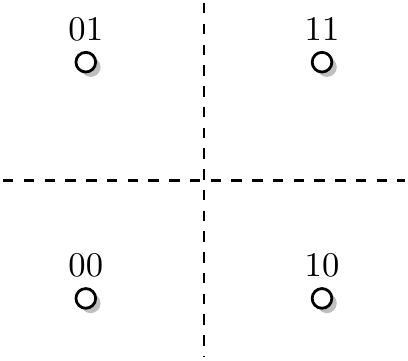}
}
\label{fig:transmitted}
}\qquad
\subfloat[Received constellation]{
\scalebox{0.75}{
\includegraphics{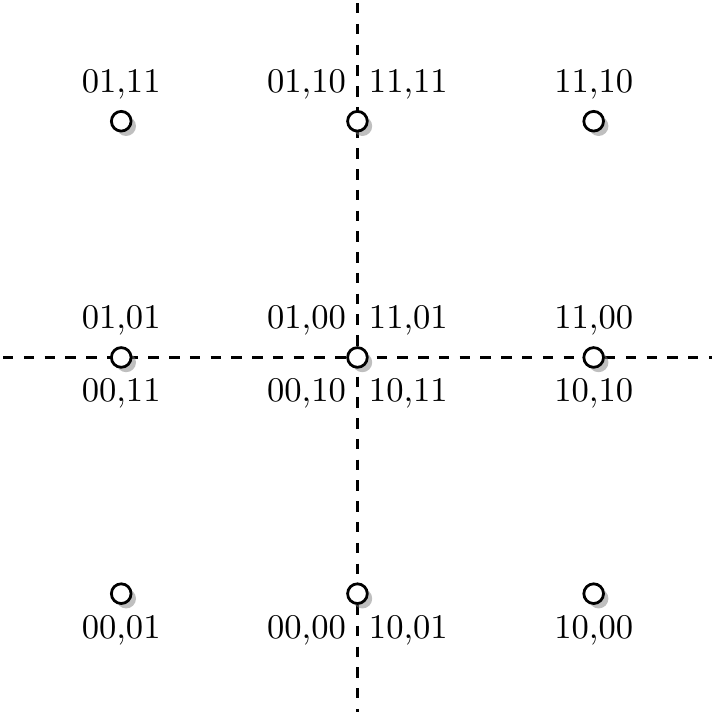}
}
\label{fig:received}
}
\caption{Transmitted and received constellations.}
\end{figure}

We begin with the role of finite rings.  As a simple starting point,
consider a PNC building block where a relay attempts to decode, at the
output of a Gaussian multiple access channel with complex-valued channel
gains, a function $f$ of messages $w_1=(w_{11},w_{12})$ and
$w_2=(w_{21},w_{22})$ sent from two transmitters, where each transmitter
uses a quaternary phase-shift-keying (QPSK) signal constellation with
Gray mapping as shown in Fig.~\ref{fig:transmitted}.  Here $w_{ij} \in
\{ 0,1 \}$.  Assume that the channel gains (at the relay) are $h_1 = 1$
and $h_2 = i$.  Then Fig.~\ref{fig:received} shows the nominal received
constellation (which is perturbed by Gaussian noise), from which the
relay must decode.  Some points in the received constellation
correspond to more than one combination of transmitted messages; for
example, $(w_1,w_2)=(01,10)$ overlaps $(w_1,w_2)=(11,11)$.  Clearly
these overlapping points  must correspond to the same value $f$, since
otherwise the relay cannot possibly form $f$ correctly.  Interestingly,
in order to achieve this, one can interpret the messages $\{w_{j1}
w_{j2}\}$ as elements in the finite ring $\ZZ_2[i] = \{ w_{j1} + w_{j2}
i \mid w_{j1}, w_{j2} \in \ZZ_2 \}$.  For example, $01$ and $10$ are
interpreted as $0+i$ and $1+0i$, respectively.  Now, consider  the
function $f: \ZZ_2[i] \times \ZZ_2[i] \to \ZZ_2[i]$ given by $f(w_1,
w_2) = w_1 + i w_2$.  In this case we have
\[
f(01, 10) = (0+i) + i(1+0i) = 0+0i = f(11,11),
\]
i.e., the points $(01,10)$ and $(11,11)$ have the same function value
$00$.  Moreover, this happens for all the overlapping points in
Fig.~\ref{fig:received} and for other channel gains as well. As such,
the finite ring $\ZZ_2[i]$ seems to be a ``good match'' for a QPSK
constellation. In fact, for \emph{every} nested-lattice-based
constellation, there is a matching finite ring, as we have shown in our
previous work \cite{FSK-submitted}.

\begin{figure}[h]
\centering\scalebox{0.75}{
\includegraphics{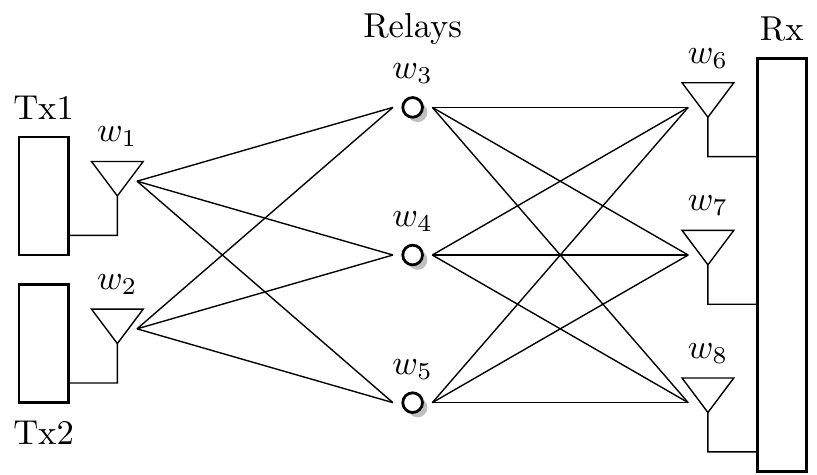}
}
\caption{A wireless relay network with three relays.}
\label{fig:relay}
\end{figure}

Next, we introduce an end-to-end matrix model that allows us to study
wireless relay networks with PNC.  Fig.~\ref{fig:relay} illustrates a
wireless relay network consisting of two transmitters, three relays, and
a single receiver (with three antennas).  Suppose that the network
employs (nested-lattice-based) PNC and the packets are over some finite
ring $R$.  Let $w_1, w_2$ be the packets at the transmitters, and let
$w_6, w_7, w_8$ be the packets at the receiver. Using PNC, each relay
node first decodes a linear combination $w_j$ ($j = 3, 4, 5$) of the
packets $w_1, w_2$, and then transmits this combination simultaneously.
Hence, we have $w_j = a_{1j} w_1 + a_{2j} w_2$ for some $a_{1j}, a_{2j}
\in R$, where $j = 3, 4, 5$.  Similarly, $w_j = a_{3j} w_3 + a_{4j} w_4
+ a_{5j} w_5$, where $j = 6, 7, 8$.  Clearly, the relation between the
transmitted packets and the received packets is given by $Y = AX$, where
\[
X = \mat{w_1 \\ w_2}, \ Y = \mat{w_6 \\ w_7 \\ w_8}
\]
and
\[
A = \mat{a_{36} & a_{46} & a_{56} \\ a_{37} & a_{47} & a_{57} \\
a_{38} & a_{48} & a_{58}}
\mat{a_{13} & a_{23} \\ a_{14} & a_{24} \\
a_{15} & a_{25}} \in R^{3 \times 2}.
\]
This gives rise to a matrix channel for the receiver.

Note that relays may sometimes introduce decoding errors.  Suppose that
the relay at the bottom of Fig.~\ref{fig:relay} makes a decoding error,
i.e., $w_5 = a_{15} w_1 + a_{25} w_2 + e$, where $e$ represents the
error packet. In this case, the receiver observes $Y = AX + Z$, where
$A$ is the same as before, and
\[
Z = \mat{a_{56} \\ a_{57} \\ a_{58}} e.
\]

The above example can be generalized to a large network.  Suppose that
we now have $n$ transmitters, $N$ relays, and $N$ receivers (each with a
single antenna).  Suppose that these receivers are connected to a
central processor (similar to the architecture of small cells or
cloud-based radio access networks).  Clearly, the central processor
observes a matrix channel $Y = AX + Z$, where $A$ is of size $N \times
n$.

To sum up, the matrix model $Y = AX + Z$ (over some finite ring)
provides a general abstraction for studying wireless relay networks with
nested-lattice-based PNC.

\section{Preliminaries}\label{sec:preli}

In this section, we present some basic results for finite chain rings
and modules and matrices over finite chain rings.  This section
establishes notation and the results that will be used later for the study
of matrix channels over finite rings;
nevertheless, this material is standard; see e.g.,
\cite{McDonald74,M84,Brown93,NS00,HL00,Nechaev.08,KS11} for more details.
To make the paper more self-contained, Appendix~\ref{sec:rings} reviews
some basic facts about rings and ideals.

\subsection{Finite Chain Rings}

All rings in this paper will be commutative with identity $1 \ne 0$.  A
ring $R$ is called a \emph{chain ring} if the ideals of $R$ satisfy the
chain condition: for any two ideals $I, J$ of $R$, either $I \subseteq
J$ or $J \subseteq I$.  If $R$ is a chain ring with finitely many
elements, then $R$ is called a \emph{finite chain ring}. Clearly, a
finite chain ring has a unique maximal ideal, and hence is local.  It is
known \cite{McDonald74} that a finite ring is a chain ring if and only
if it is a local principal ideal ring (PIR); thus, in a finite chain
ring, all ideals are principal.  Examples of finite chain rings include
$\ZZ_{p^n}$ (the ring of integers modulo $p^n$ where $p$ is a prime) and
Galois rings.

Let $R$ be a finite chain ring, and let $\pi \in R$ be any
generator of the maximal ideal of $R$.  Then $R/\langle \pi
\rangle$ is the residue field of $R$.  It can be shown (see,
e.g., \cite{McDonald74}) that every ideal $I$ of $R$,
including the zero ideal $\langle 0 \rangle$, is generated by
a power of $\pi$, i.e., $I = \langle \pi^{l} \rangle$ for some
$l \ge 0$.  It follows that $\pi$ is nilpotent; we denote by
$\len$ the \emph{nilpotency index} of $\pi$, i.e., the smallest
positive integer such that $\pi^{\len} = 0$.  There are, then,
exactly $\len + 1$ distinct ideals of $R$, namely, $R=\langle
\pi^0 \rangle, \langle \pi^1 \rangle, \ldots, \langle
\pi^{\len} \rangle = \{ 0 \} $ which form a chain (with
respect to set inclusion):
\[
  R = \langle \pi^0 \rangle \supset \langle \pi^1 \rangle
      \supset  \cdots \supset \langle \pi^{\len -1} \rangle
      \supset \langle \pi^{\len} \rangle = \{ 0 \}.
\]
Thus, $\len$ is often called the \emph{chain length} of $R$.
We refer to $R$ as a $(q, \len)$ chain ring if $R$ has a
residue field of size $q$ and a chain length of $\len$.

\begin{exam}\label{ex:fcr}
The ideals of $\ZZ_8$ form a chain with respect to set
inclusion:
\[
  R = \langle 1 \rangle \supset \langle 2 \rangle
      \supset \langle 4 \rangle
      \supset \langle 0 \rangle = \{ 0 \}.
\]
Thus, $\ZZ_8$ is a finite chain ring with chain length
$\len=3$.  Since the residue field $\ZZ_8 / \langle 2 \rangle$
is isomorphic to $\FF_2$, $\ZZ_8$ is a $(2,3)$ chain ring.
\end{exam}

Now let $\calR(R, \pi) \subseteq R$ be a complete set of
residues with respect to $\pi$ and, without loss of
generality, assume that $0 \in \calR(R, \pi)$.  Every element
$a \in R$ then has a unique representation, called the {\em
$\pi$-adic decomposition} of $a$ (with respect to $\calR(R,
\pi)$), in the form
\begin{equation}\label{eq:pi-adic}
   a = a_0 + a_1 \pi + \cdots + a_{\len -1} \pi^{\len -1},
\end{equation}
where $a_0, \ldots, a_{\len - 1} \in \calR(R, \pi)$.  It
follows from the uniqueness of \eqref{eq:pi-adic} that the
size of $R$ is $q^\len$, i.e., the number of elements in a
$(q, \len)$ chain ring is $q^\len$.  Thus, like a finite
field, a finite chain ring has a cardinality that is an
integer power of a prime number.

The \emph{degree} of a nonzero element $a_0 + a_1 \pi + \cdots
+ a_{\len-1}\pi^{\len-1} \in R$, denoted by $\deg(a)$, is
defined as the \emph{least} index $j$ for which $a_j \ne 0$.
By convention, the degree of $0$ is defined as $\len$.  All
elements of the same degree are associates in $R$.  Further,
$a$ divides $b$ if and only if $\deg(a) \le \deg(b)$.
Finally, $\deg(a+b) \geq \min\{ \deg(a), \deg(b) \}$, i.e.,
adding two elements never results in an element of lower
degree.

\begin{exam}
Let $\calR(\ZZ_8, 2) = \{ 0, 1 \}$.  The $2$-adic
decomposition of $5 \in \ZZ_8$ is $5 = 1 + 0 \cdot 2 + 1 \cdot
2^2$.  The elements in $\ZZ_8$ of degree $0$ (respectively,
$1$, $2$, and $3$) are $\{ 1, 3, 5, 7\}$ (respectively, $\{ 2,
6 \}$, $\{ 4 \}$, and $\{ 0 \}$).
\end{exam}

Finally, we present two methods for constructing finite chain
rings.

If $R$ is itself a $(q, \len)$ chain ring with maximal ideal
$\langle \pi \rangle$, then the quotient $R / \langle \pi^l
\rangle$ ($0 < l < \len$) is a $(q, l)$ chain ring.  This
method constructs new finite chain rings from existing ones.

If $T$ is a principal ideal domain (PID), and $p$ is a prime in $T$, then $T / \langle
p \rangle$ is a field, since $\langle p \rangle$ is a maximal
ideal of $T$.  Let $q$ be the size of $T / \langle p \rangle$
and suppose that $q$ is finite.  Then the quotient $T /
\langle p^l \rangle$ is a $(q, l)$ ($l > 0$) chain ring.  This
method constructs finite chain rings from PIDs.

\subsection{Modules over Finite Chain Rings}

A module is to a ring as a vector space is to a field.  More
formally, an $R$-module $M$ is an abelian group $(M, +)$
together with an action of $R$ on $M$ satisfying the
following conditions for all $m,n \in M$ and for all $a,b \in
R$:
\begin{enumerate}
\item $1 m = m$ and $(a b) m = a( b m)$
\item $(a + b) m = am + bm$
\item $a(m + n) = a m + an$
\end{enumerate}

When $R$ is a finite chain ring, an $R$-module is always
isomorphic to a direct product of various ideals of $R$;  this
structure can be described by a ``shape.'' An \emph{$\len$-shape}
$\mu = (\mu_1, \mu_2, \ldots, \mu_\len)$ is simply a sequence of
non-decreasing non-negative integers, i.e., $0 \le \mu_1 \le
\mu_2 \le \cdots \le \mu_\len$.  We denote by $|\mu|$ the sum
of its components, i.e., $|\mu| = \sum_{i = 1}^\len \mu_i$.
For later notational convenience, we define the ``zeroth
component'' of a shape as $\mu_0=0$.

An $\len$-shape $\kappa = (\kappa_1, \ldots, \kappa_{\len})$
is said to be a \emph{subshape} of $\mu = (\mu_1, \ldots,
\mu_\len)$, written $\kappa \preceq \mu$, if $\kappa_i \le
\mu_i$ for all $i = 1, \ldots, \len$.  Thus, for example, $(1,
1, 3) \preceq (2, 4, 4)$.  The number of subshapes of the
$\len$-shape $(m, \ldots, m)$ is given by $\binom{m +
\len}{\len}$, which implies that the number of subshapes of
$\mu = (\mu_1, \ldots, \mu_\len)$ is upper-bounded by
$\binom{\mu_s + \len}{\len}$.

Two $\len$-shapes can be added together to form a new
$\len$-shape simply by adding componentwise.  Thus, for
example, $(1, 1, 3) + (2, 4, 4) = (3, 5, 7)$. Also, for a
shape $\mu = (\mu_1, \ldots, \mu_\len)$ and a positive integer
$m$ we define $\mu/m = (\mu_1/m, \ldots, \mu_\len/m)$ (which
is an $\len$-tuple, but not necessarily a shape).  For
convenience, we will sometimes identify the integer $t$ with
the $\len$-shape $(t, \ldots, t)$.  Thus, for example, $\mu
\preceq t$ means $\mu_i \le t$ for all $i$, $\kappa = t$ means
$\kappa_i = t$ for all $i$, and $\mu - t = (\mu_1 - t, \ldots,
\mu_\len - t)$, assuming $ t \preceq \mu$.

Let $R$ be a $(q,s)$ chain ring with maximal ideal $\langle
\pi \rangle$.  For any $\len$-shape $\mu$, we define the
$R$-module $R^{\mu}$ as
\ifCLASSOPTIONonecolumn
\begin{equation}\label{eq:shape-def}
R^{\mu} \triangleq
  \underbrace{\langle 1\rangle\times\cdots\times\langle 1\rangle}_{\mu_1}
  \times
  \underbrace{\langle\pi\rangle\times\cdots\times\langle\pi\rangle}_{\mu_{2}-\mu_1}
  \times \cdots
  \times
  \underbrace{\langle \pi^{\len -1} \rangle \times \cdots \times \langle \pi^{\len -1} \rangle}_{\mu_{\len} - \mu_{\len - 1}}.
\end{equation}
\else  
\begin{multline}\label{eq:shape-def}
  R^{\mu} \triangleq
  \underbrace{\langle 1\rangle\times\cdots\times\langle 1\rangle}_{\mu_1}
  \times
  \underbrace{\langle\pi\rangle\times\cdots\times\langle\pi\rangle}_{\mu_{2}-\mu_1}
  \times \cdots \\
  \times
  \underbrace{\langle \pi^{\len -1} \rangle \times \cdots \times \langle \pi^{\len -1} \rangle}_{\mu_{\len} - \mu_{\len - 1}}.
\end{multline}
\fi
Since a positive integer $t$ is identified with the shape
$(t,\ldots,t)$, it is indeed true that $R^t$ denotes the
$t$-fold Cartesian product of $R$ with itself.

The module $R^{\mu}$ can be viewed as a collection of
$\mu_\len$-tuples whose components are drawn from $R$ subject
to certain constraints imposed by $\mu$.  Specifically, while
the first $\mu_{1}$ components can be any element of $R$, the
next $\mu_2 - \mu_1$ components must be multiples of $\pi$,
and so on.  Since each ideal $\langle \pi^i \rangle$ in
(\ref{eq:shape-def}) contains $q^{s - i}$ elements ($0 \le i <
s$), it follows that the size of $R^{\mu}$ is $|R^{\mu}| =
q^{|\mu|}$.

\begin{exam}
Let $R = \ZZ_8$, and let $\mu = (2, 4, 4)$. Then
\[
R^{\mu} = \underbrace{\langle 1 \rangle \times \langle 1 \rangle}_2 \times
          \underbrace{\langle 2 \rangle \times \langle 2 \rangle}_{4 - 2}.
\]
Note that the first two components of $R^{\mu}$ can each be
chosen in $2^3$ ways, while the last two components can each
be chosen in only $2^2$ ways.  Hence, the size of $R^{\mu}$ is
$2^{10}$.
\end{exam}

For every $\len$-shape $\mu$, $R^{\mu}$ is a finite
$R$-module.  Conversely, the following theorem establishes
that every finite $R$-module is isomorphic to $R^{\mu}$ for
some unique $\len$-shape $\mu$.
\begin{thm}\cite[Theorem~2.2]{HL00}
For any finite $R$-module $M$ over a $(q,s)$ chain ring $R$,
there is a unique $\len$-shape $\mu$ such that $M \cong
R^{\mu}$.
\label{thm:ModuleShape}
\end{thm}

We call the unique shape $\mu$ given in
Theorem~\ref{thm:ModuleShape} \emph{the shape of} $M$, and write $\mu
= \shape M$.\footnote{Some authors (like Honold \emph{et
al.}~\cite{HL00}) use a different convention and define the
shape of an $R$-module to be the conjugate (in the
integer-partition-theoretic sense) of the shape as defined in
this paper.} It is known \cite{HL00} that if $M'$ is a
submodule of $M$, then $\shape M' \preceq \shape M$, i.e., the
shape of a submodule is a subshape of the module.  It is also
known \cite{HL00} that the number of submodules of $R^{\mu}$
whose shape is $\kappa$ is given by
\begin{equation}
    \submodule{\mu}{\kappa}
     = \prod_{i = 1}^\len q^{(\mu_i - \kappa_i) \kappa_{i-1}}
       \gauss{\mu_i - \kappa_{i-1}}{\kappa_i - \kappa_{i-1}},
    \label{eqn:ModuleCount}
\end{equation}
where
\[
\gauss{m}{k} \triangleq \prod_{i = 0}^{k-1} \frac{q^m - q^i}{q^k - q^i}
\]
is the Gaussian coefficient.  In particular, when the chain
length $\len = 1$, $R$ becomes the finite field $\FF_q$ of $q$
elements, and $\submodule{\mu}{\kappa}$ becomes
$\gauss{\mu_1}{\kappa_1}$, which is the number of
$\kappa_1$-dimensional subspaces of $\FF_q^{\mu_1}$.

\subsection{Matrices over Finite Chain Rings}\label{sec:matrix}

We turn now to matrices over finite chain rings.  Let $R$ be a
$(q, \len)$ chain ring with maximal ideal $\langle \pi
\rangle$.  The set of all $n \times m$ matrices with entries
from $R$ will be denoted by $R^{n \times m}$.  If $A \in R^{n
\times m}$, we denote by $A[i, j]$ the entry of $A$ in the
$i$th row and $j$th column, where $1 \leq i \leq n$ and $1
\leq j \leq m$. We will let $A[i_1{:}i_2, j_1 {:} j_2]$
denote the submatrix of $A$ formed by rows $i_1$ to $i_2$ and
by columns $j_1$ to $j_2$, where $1 \le i_1 \le i_2 \le n$ and
$1 \le j_1 \le j_2 \le m$.  Finally, we will let $A[i,{:}]$
denote the $i$th row of $A$ and $A[{:},j]$ denote the $j$th
column $A$.

A square matrix $U \in R^{n \times n}$ is \emph{invertible} if
$UV = VU = I_n$ for some $V \in R^{n \times n}$, where $I_n$
denotes the $n \times n$ identity matrix.  The set of
invertible matrices in $R^{n \times n}$, denoted as
$\GL_n(R)$, forms a group---the so-called \emph{general linear
group}---under matrix multiplication.

Two matrices $A, B \in R^{n \times m}$ are said to be
\emph{left-equivalent} if there exists a matrix $U \in
\GL_n(R)$ such that $UA = B$. Two matrices $A, B \in R^{n
\times m}$ are said to be \emph{equivalent} if there exist
matrices $U \in \GL_n(R)$ and $V \in \GL_m(R)$ such that $U A
V = B$.

A matrix $D \in R^{n \times m}$ is called a \emph{diagonal
matrix} if $D[i, j] = 0$ whenever $i \ne j$. A diagonal matrix
$D$, which need not be square, can be written as $D =
\diag(d_1, \ldots, d_r)$, where $r = \min\{n, m\}$, and $d_i =
D[i, i]$ for $i = 1, \ldots, r$.

Let $A \in R^{n \times m}$. A diagonal matrix $D= \diag(d_1,
\ldots, d_r) \in R^{n \times m}$ ($r = \min\{n, m\}$) is
called a \emph{Smith normal form} of $A$, if $D$ is equivalent
to $A$ and $d_1 \mid d_2 \mid \cdots \mid d_r$ in $R$.  It is
known \cite{Brown93} that every matrix over a PIR (in
particular, a finite chain ring) has a Smith normal form whose
diagonal entries are unique up to equivalence of associates.
In this paper, we shall require the diagonal entries $d_1,
\ldots, d_r$ in the Smith normal form $D$ to be powers of
$\pi$, i.e.,
\[
 (d_1, \ldots, d_r) = (\pi^{l_1}, \ldots, \pi^{l_r}),
\]
where $0\le l_1 \le \ldots \le l_r \le \len$ since $d_1 \mid
d_2 \mid \cdots \mid d_r$.  With this constraint, once $\pi$
is fixed, every matrix $A \in R^{n \times m}$ has a unique
Smith normal form.

\begin{exam}\label{ex:Smith}
Consider the two matrices
\[
    A = \mat{4 & 6 & 2 & {1}\\ 0 & 0 &0 & 2\\ 2 & 4 & 6 &1 \\2 & 0 & 2 & 1}, \
    S =  \mat{{1} & 0 & 0 & 0\\
      0 & {2} & 0 & 0 \\ 0 & 0 & {4} & 0\\ 0 & 0 & 0 & 0}
\]
over $\ZZ_8$. It is easy to check that
\[
    A  = \mat{1 & 2 & 0 & 0 \\ 2 & 0 & 1 & 0\\
    1 & 1 & 0 & 0\\ 1 & 1 & 1 & 1}
    \mat{{1} & 0 & 0 & 0\\
          0 & {2} & 0 & 0 \\ 0 & 0 & {4} & 0\\ 0 & 0 & 0 & 0}
         \mat{0 & 2 & 2 & 1\\ 1 & 1 & 2 & 0\\
         0 & 1 & 1 & 0\\ 0 & 0 & 1 & 0} = U S V.
\]
Since $U$ and $V$ are invertible, $S$ is equivalent to $A$.
Since the diagonal entries of $S$ satisfy $1 \mid 2 \mid 4
\mid 0$ in $\ZZ_8$, $S$ is the Smith normal form of $A$.
\end{exam}

For any $A \in R^{n \times m}$, we denote by $\row A$
and $\col A$ the row span and column span of $A$, respectively.
By using the Smith normal form, it is easy to see that
$\row A$ is isomorphic, as an $R$-module, to
$\col A$.
It is also easy to see that
left-equivalent matrices have identical row spans and
equivalent matrices have isomorphic row spans.

The \emph{shape} of a matrix $A$ is defined as the shape of
the row span of $A$, i.e.,
\[
    \shape A = \shape(\row A).
\]
Clearly, $\shape A = \shape(\col A)$.
Moreover, $\shape A = \mu$ if and only if the Smith normal form of $A$
is given by
\[
\diag(\underbrace{1, \ldots, 1}_{\mu_1}, \underbrace{\pi, \ldots, \pi}_{\mu_2 - \mu_1},
\ldots, \underbrace{\pi^{s-1}, \ldots, \pi^{s-1}}_{\mu_s - \mu_{s-1}},
\underbrace{0, \ldots, 0}_{r - \mu_s}),
\]
where $r = \min\{n, m \}$.
In particular, a matrix $U \in R^{n \times n}$ is invertible if and only if
$\shape U = (n, \ldots, n)$.

\begin{exam}
Since $D = \diag(1, 2, 4, 0)$ is the Smith normal form of $A$ in
Example~\ref{ex:Smith}, $\shape A = (1, 2, 3)$.
\end{exam}

As one might expect, matrix shape has a number of properties
similar to matrix rank.

\begin{prop}\label{prop:shape-properties}
Let $A \in R^{n \times m}$ and $B \in R^{m \times k}$. Then
\begin{enumerate}
\item $\shape A = \shape A^T$, where $A^T$ is the transpose of $A$.
\item For any $P \in \GL_n(R), Q \in \GL_m(R)$, $\shape A = \shape PAQ$.
\item $\shape AB \preceq \shape A$, $\shape AB \preceq \shape B$.
\item For any submatrix $C$ of $A$, $\shape C \preceq \shape A$.
\item $\shape A \preceq \min\{ n, m \}$.
\end{enumerate}
\end{prop}
\begin{proof}
1) Since $\row A \cong \col A$, we have $\row A \cong \row
A^T$.  Hence, $\shape A = \shape A^T$.  2) Since $A$ is
equivalent to $P A Q$ for any invertible $P$ and $Q$, $\shape
A = \shape  PAQ$.  3) Since $\row AB$ is a submodule of $\row
B$, we have $\shape AB \preceq \shape B$.  Similarly, since
$\col AB$ is a submodule of $\col A$, we have $\shape AB
\preceq \shape A$.  4) Note that any submatrix $C$ of
$A$ is equal to $E_1 A E_2$ for some $E_1 \in R^{k \times n}$
(selecting $k$ rows) and $E_2 \in R^{m \times l}$ (selecting $l$
columns).  Hence, $\shape C = \shape E_1 A E_2 \preceq \shape
A$. 5) Since the Smith normal form of $A$ has at most $\min\{ n, m \}$
nonzero diagonal entries, we have $\shape A \preceq \min\{ n, m \}$.
\end{proof}

For convenience, we say a matrix $A \in R^{n \times m}$
have \emph{rank} $t$,
if $\shape A = t$. Note that the rank of a matrix is not
always defined.
A matrix $A \in R^{n \times m}$ is called \emph{full rank} if
$\rank A = \min\{ n, m\}$.  A matrix $A \in R^{n \times m}$ is
called \emph{full row rank} if $\rank A = n$ (which requires
$n \le m$). The number of full-row-rank matrices in $R^{n
\times m}$ is $q^{s n m} \prod_{i = 0}^{n - 1} (1 - q^{i-m})$.
A matrix is \emph{full column rank} if
its transpose is full row rank. Full-column-rank matrices
have the following property.

\begin{lem}\label{lem:column-rank}
Let $A$ be a full-column-rank matrix. Then $A B$ is a zero matrix
if and only if $B$ is a zero matrix.
\end{lem}
\begin{proof}
The ``if" part is trivial, so we turn to the ``only if" part. Let $A \in R^{n \times m}$.
Suppose that $AB = 0$ for some matrix $B \in R^{m \times k}$. We will show that $B$
is a zero matrix. Since $A$ is full column rank,
its Smith normal form $S$ must have the form
\[
  S = \mat{I_m \\ 0_{(n - m) \times m}}
\]
and $A = U S V$ for some invertible matrices $U$ and $V$. Thus, we have
\[
  AB = U \mat{I \\ 0} V B = 0,
\]
which implies $B = 0$.
\end{proof}

\section{Row Canonical Form}\label{sec:canonical}

The main algebraic tools for studying matrix channels over finite fields
include Gaussian elimination and reduced row echelon forms.
The generalization of these tools to finite chain rings is,
however, not straightforward.
Consider the $3 \times 4$ matrix
\[
A = \mat {2 & 1 & 1 & 2\\6 & 3 & 7 & 2 \\ 6 & 7 & 1 & 0}
\]
over $\ZZ_{8}$. On the one hand, we have
\[
\mat{1 & 0 & 0\\1 & 1 & 0 \\ 1 & 0  & 1} A = \mat {2 & 1 & 1 & 2\\ 0 & 4 & 0 & 4 \\ 0 & 0 & 2 & 2}.
\]
On the other hand, we have
\[
\mat{1 & 0 & 0\\1 & 0 & 1 \\ 7 & 1 & 2} A = \mat {2 & 1 & 1 & 2 \\ 0 & 0 & 2 & 2\\ 0 & 0 & 0 & 0}.
\]
In both cases we have transformed $A$ to echelon form using
elementary row operations.
Recall that, over finite fields, the rank of a matrix is
precisely the number of nonzero rows in its echelon form.
This property, however, does not hold for matrices over finite chain rings.

To address this difficulty, several possible generalizations
of reduced row echelon forms have been proposed in the literature,
including the Howell form \cite{How86, Stor00}, the matrix canonical form \cite{F55, McDonald74},
and the $p$-basis \cite{VSR96}.
In this section, we will describe a row canonical form that is particularly
suitable for studying matrix channels over finite chain rings.
This row canonical form is essentially the same as
the reduced row echelon form defined in Kiermaier's thesis \cite[Definition~2.2.2]{Kiermaier12} (written in German),
which itself is
a variant of the matrix canonical form in \cite[p.~329, Exercise XVI.7]{McDonald74}.
It appears that the key idea behind these forms was proposed by Fuller \cite{F55}
based on an earlier result of Birkhoff \cite{B35}.
We provide in this section a new elementary proof for the existence
and uniqueness of the row canonical form.

Throughout this section, $R$ is a $(q, \len)$ chain ring with
maximal ideal $\langle \pi \rangle$. We fix a complete set of
residues $\calR(R,\pi)$ (including $0$), i.e., a
representation of the residue field $R/\langle \pi \rangle$,
and, for $1< l < \len$, we choose the complete set of residues
for $\pi^l$ as
\[
  \calR(R, \pi^l) =  \left\{ \sum_{i=0}^{l-1} a_i \pi^{i} \colon
   a_0, \ldots, a_{l-1} \in  \calR(R, \pi) \right\}.
\]
Finally, we set $\calR(R, \pi^0) = \{ 0 \}$.

\subsection{Definitions}
We start with a few definitions.

Let $A$ be matrix with entries from $R$. The $i$th row of
$A$ is said to occur \emph{above} the $(i')$th row of $A$ (or
the $(i')$th row occurs \emph{below} the $i$th row) if $i <
i'$.  Similarly the $j$th column of $A$ is said to occur
\emph{earlier} than the $(j')$th column (or the $(j')$th
column occurs \emph{later} than the $j$th column) if $j < j'$.
This terminology extends to the entries of $A$:  $A[i,j]$ is
above $A[i',j']$ if $i<i'$ and $A[i,j]$ is earlier than
$A[i',j']$ if $j<j'$.  If $P$ is some property obeyed by at
least one of the entries in the $i$th row of $A$, then the
\emph{first} entry in row $i$ with property $P$ occurs earlier
than every other entry in row $i$ having property $P$.

The \emph{pivot} of a nonzero row of a matrix is the first
entry among the entries having least degree in that row.  For
example, $6$ and $2$ are the entries of least degree in the
row $[0 ~ 4 ~ 6 ~ 2]$ over $\ZZ_8$, and $6$ occurs earlier.
Thus, $6$ is the pivot of the row $[0 ~ 4 ~ 6 ~ 2]$.  Note
that the pivot of a row is not necessarily the first nonzero
entry of the row.

\begin{defi}
A matrix $A$ is in \emph{row canonical form} if it satisfies
the following conditions.
\begin{enumerate}
\item Nonzero rows of $A$ are above any zero rows.

\item If $A$ has two pivots of the same degree, the one that
occurs earlier is above the one that occurs later.  If $A$ has
two pivots of different degree, the one with smaller degree is
above the one with larger degree.

\item Every pivot is of the form $\pi^l$ for some $l  \in \{
0, \ldots, \len - 1\}$.

\item For every pivot (say $\pi^l$), all entries below and in
the same column as the pivot are zero, and all entries above and
in the same column as the pivot are elements of $\calR(R, \pi^l)$.
\end{enumerate}
\end{defi}

\begin{exam}\label{ex:canonical-form}
Consider the matrix
\[
    A = \mat{0 & 2 & 0 & \bar{1} \\
             \bar{2} & 2 & 0 & 0 \\
              0 & 0 &\bar{2} & 0 \\
             0 & \bar{4} & 0 & 0 \\
                   0 & 0 & 0 & 0}
\]
over $\ZZ_8$ with $\pi=2$ and $\ZZ_8/\langle 2 \rangle = \{
0,1 \}$, in which the pivots have been identified with an
overline. Clearly, $A$ satisfies all of the conditions to be
in row canonical form.
\end{exam}

The following facts follow immediately from the definition
of row canonical form.

\begin{prop}
\label{prop:rcfprops}
Let $A \in R^{n \times m}$ be a matrix in row canonical form, let $p_k$
be the pivot of the $k$th row,
let $c_k$ be the index of the column containing $p_k$.
(If the $k$th row is zero, let $p_k=0$ and $c_k=0$.)
Let $d_k = \deg(p_k)$,
and let $w=(w_1,\ldots,w_m)$ be
an arbitrary element of $\row A$.
\begin{enumerate}
\item Any column of $A$ contains at most one pivot.
\item If $A$ has more than one row, deleting a row of $A$
results in a matrix also in row canonical form.
\item $i \geq k$ implies $\deg(A[i,j]) \geq d_k $.
\item ($i \geq k$ and $j < c_k$) or ($i > k$ and $j \leq c_k$)
implies $\deg(A[i,j]) > d_k$.
\item $p_1$ divides $w_1, w_2, \ldots, w_m$.
\item $j < c_1$ implies $\deg(w_j) > d_1$.
\end{enumerate}
\end{prop}
The proof is provided in Appendix~\ref{sec:canonical-proofs}.
For any $A \in R^{n \times m}$, we say a matrix $B \in R^{n
\times m}$ is a \emph{row canonical form of $A$}, if (i) $B$
is in row canonical form, and (ii) $B$ is left-equivalent to
$A$.  We will show that any $A \in R^{n \times m}$ has a
unique row canonical form.  For this reason, we denote by
$\RCF(A)$ \emph{the} row canonical form of $A$.

\subsection{Existence and Uniqueness}
First, we demonstrate the existence of a row canonical form for
any matrix $A$ by presenting a simple algorithm that performs
elementary row operations to reduce $A$ into row canonical
form.  Here, the allowable elementary row operations (over
$R$) are:
\begin{itemize}
\item Interchange two rows.
\item Add a multiple of one row to another.
\item Multiply a row by a unit in $R$.
\end{itemize}
Each of these operations is invertible, and so a matrix
obtained from $A$ by any sequence of these operations will
have the same row span as $A$.

The algorithm proceeds in a series of steps.  In the $k$th
step, the algorithm selects the $k$th pivot, moves it to the
$k$th row, and uses elementary row operations to reduce into
row canonical form the submatrix consisting of the top $k$
rows.  The pivot selection procedure operates on any given set
of rows.  If the rows are all zero, the procedure should
return with the result that no pivot can be found.  Otherwise,
among all entries of least degree in the given rows, an entry
must be chosen that occurs as early as possible.  This entry
must certainly be the pivot of its row.  The procedure should
return the row and column index of the selected element.

Now we are ready to describe the algorithm in detail.  In step
$k=1$, apply pivot selection to all of the rows of $A$.  If no
pivot can be found, then $A$ is a zero matrix, and is already
in row canonical form.  Otherwise, we call this pivot the
\emph{first pivot} and place it in the first row by an
interchange of rows (if necessary).  If this pivot is not of
the form $\pi^l$ ($l = 0, \ldots, \len - 1$), we multiply the
first row by a suitable unit so that the first pivot is a
power of $\pi$.  Note that nonzero entries in the same column
below the first pivot have degrees no less than the pivot,
which means that they are all multiples of the first pivot.
By a sequence of elementary row operations, these entries can
be cancelled, so that we arrive at a matrix, say $A_1$, in
which the first row is in row canonical form and all entries
in the same column below the first pivot are zero.  We can now
increment $k$ and proceed to the next step.

For $k \geq 2$, we apply pivot selection to the rows of
$A_{k-1}$, excluding the first $k-1$ rows.  If no pivot can be
found, then the remaining rows are all zero and $A_{k-1}$ is
in row canonical form.  Otherwise we call this pivot the $k$th
pivot and place it in the $k$th row by an exchange of rows (if
necessary).  As in the first step, if this pivot is not an
integer power of $\pi$, we multiply the $k$th row by a
suitable unit so that the $k$th pivot is a power of $\pi$, say
$\pi^l$.  Nonzero entries in the same column below the $k$th
pivot can be cancelled using elementary row operations.  A
nonzero entry, say $a$, in the same column above the $k$th
pivot has $\pi$-adic decomposition
\begin{align*}
a & =  a_0 + \cdots + a_{\len - 1}\pi^{\len - 1} \\
  & =  a_0 + \cdots + a_{l-1} \pi^{l-1} + \pi^l (a_l + \cdots
        + a_{\len-1}\pi^{\len - l - 1}).
\end{align*}
Thus by subtracting $(a_l + \cdots + a_{\len - 1}\pi^{\len - l
- 1})$ times the $k$th row from the row containing $a$, we
change $a$ to $a_0 + \cdots + a_{l-1} \pi^{l-1} \in \calR(R,
\pi^l)$, without affecting the pivot of that row.  Reducing
all nonzero entries in the same column as the $k$th pivot in
this way, we arrive at a matrix, say $A_k$, in which the top
$k$ rows are in row canonical form and all entries in the same
column below the first, second, \ldots, $k$th pivots are zero.

The above algorithm stops when no more pivots can be found.
Note that, at the end of the $k$th step, the matrix $A_k$ is
left-equivalent to $A$ and the submatrix formed by the top $k$
rows of $A_k$ is in row canonical form.  It follows that the
final matrix must be in row canonical form.

Therefore, we have the following result.
\begin{prop}
For any $A \in R^{n \times m}$, the algorithm described above
computes a row canonical form of $A$.
\end{prop}
A simple count shows that this algorithm requires
\[
\calO(nm \min\{ n, m \})
\]
basic operations over $R$.

\begin{exam}\label{ex:canonical}
Consider the matrix
\[
A = \mat{4 & 6 & 2 & \bar{1}\\
         0 & 0 & 0 & 2\\
         2 & 4 & 6 & 1 \\
         2 & 0 & 2 & 1}
\]
over $\ZZ_8$.  There are three $1$s in the last column of $A$,
namely, $A[1, 4]$, $A[3, 4]$ and $A[4,4]$, which are the
elements of least degree in $A$.  We can choose any of them as
the first pivot.  Here, we choose $A[1, 4]$ (indicated by an
overline).  After some elementary row operations, we can make
the entries below the pivot zero to obtain
\[
A_1 = \mat{ 4 & 6 & 2 & {1} \\
            0 & 4 & 4 & 0   \\
      \bar{6} & 6 & 4 & 0   \\
            6 & 2 & 0 & 0 }.
\]
Now consider the submatrix formed by omitting the first row of
$A_1$.  There are four entries of least degree, namely,
$A_1[3,1] = 6$, $A_1[3, 2] = 6$, $A_1[4, 1] = 6$, and $A_1[4,
2] = 2$, among which $A_1[3,1]$ and  $A_1[4, 1]$ are valid
choices for the second pivot.  Here, we choose $A_1[3,1]$
(indicated by an overline).  We interchange the second row and
third row of $A_1$, and then multiply the new second row by
$3$, obtaining
\[
A_1' = \mat{ 4 & 6 & 2 & {1}\\
       \bar{2} & 2 & 4 &  0 \\
             0 & 4 & 4 &  0 \\
             6 & 2 & 0 & 0 }.
\]
By some elementary row operations, we can make the entries
below the second pivot zero.  After that, we subtract $2$
times the second row from the first row, obtaining
\[
A_2 = \mat{ 0 & 2 & 2 & {1} \\
          {2} & 2 & 4 & 0   \\
            0 & \bar{4} & 4 &0 \\
            0 & 4 & 4 & 0 }.
\]
Clearly, the submatrix formed by the top two rows of $A_2$ is
in row canonical form.  Next, consider the submatrix formed by
omitting the top two rows of $A_2$.  We choose the entry
$A_2[3, 2]$ (indicated by an overline) as the third pivot. We
subtract the third row from the fourth row and obtain
\[
A_3 = \mat{ 0 & 2 & 2 & \bar{1} \\
      \bar{2} & 2 & 4 & 0       \\
            0 & \bar{4} & 4 &0  \\
            0 & 0& 0 & 0 }.
\]
Clearly, the submatrix formed by the top three rows of $A_3$
is in row canonical form (with all the pivots indicated).
Since no more pivots can be found, our algorithm outputs
$A_3$, which is indeed in row canonical form.
\end{exam}

As expected, the row canonical form is unique.
\begin{prop}\label{prop:uniqueness}
For any $A \in R^{n \times m}$, the row canonical form of $A$ is unique.
\end{prop}
The proof is provided in Appendix~\ref{sec:canonical-proofs}.

\section{Matrices under Row Constraints}
\label{sec:matrix-constraint}

In this section, we study a class of matrices in $R^{n
\times m}$ whose rows are constrained to be elements of
$R^{\mu}$.
We provide
several new counting results
and a construction of principal row canonical forms
for this class of matrices. These results
are of primary importance to our
study of capacities and coding schemes in later sections.

\subsection{$\pi$-adic Decomposition}

Let $R^{n \times \mu}$ denote the set of matrices in $R^{n
\times m}$ whose rows are elements of $R^{\mu}$.  Then the
size of $R^{n \times \mu}$ is
\begin{equation}\label{eq:count-all-matrices}
    | R^{n \times \mu} | = |R^{\mu}|^n = q^{n|\mu|},
\end{equation}
since there are $|R^{\mu}| = q^{|\mu|}$ choices for each row.
Taking the logarithm on both sides of
(\ref{eq:count-all-matrices}), we obtain
\begin{equation}\label{eq:log-ambient}
        \log_q |R^{n \times \mu}| = n |\mu|.
\end{equation}

\begin{figure}[t]
  \centering
\ifCLASSOPTIONonecolumn
  \setlength{\tabcolsep}{1pt}
  \begin{tabular}{ccc}
     \includegraphics{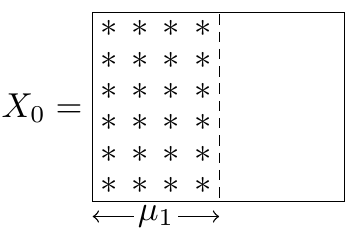} &
     \includegraphics{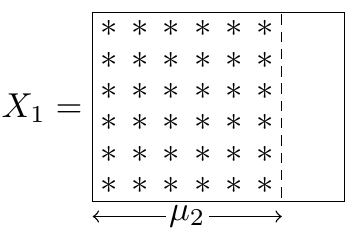} &
     \includegraphics{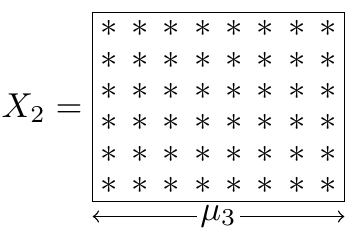}
   \end{tabular}
\else
  \scalebox{0.75}{\setlength{\tabcolsep}{1pt}\begin{tabular}{ccc}
     \includegraphics{fig3a} &
     \includegraphics{fig3b} &
     \includegraphics{fig3c}
   \end{tabular}}
\fi
\caption{Illustration of a $\pi$-adic decomposition for
         $\len = 3$ and $\mu=(4, 6, 8)$.}
\label{fig:constraints}
\end{figure}

Every matrix $X \in R^{n \times \mu}$ can be constructed based
on its $\pi$-adic decomposition
\[
   X = X_0 + \pi X_1 + \cdots + \pi^{\len - 1} X_{\len - 1},
\]
with each auxiliary matrix $X_i$ ($i = 0, \ldots, \len - 1$)
satisfying:
\begin{enumerate}
\item $X_i[1{:} n, 1 {:} \mu_{i+1}]$ is an arbitrary
matrix over $\calR(R, \pi)$, and
\item all other entries in $X_i$ are zero.
\end{enumerate}

The construction is illustrated in Fig.~\ref{fig:constraints}.
Clearly, this construction provides a one-to-one mapping from
sequences of $n |\mu|$ $q$-ary symbols to matrices in $R^{n
\times \mu}$.

\subsection{Row Canonical Forms in $\calT_{\kappa}(R^{n\times\mu})$}
\label{sec:construction}

Let $\calT_{\kappa}(R^{n \times \mu})$ denote the set of
matrices in $R^{n \times \mu}$  whose shape is $\kappa$.  Then
$|\calT_{\kappa}(R^{n \times \mu})|=0$ unless $\kappa \preceq
n$ and $\kappa \preceq \mu$ (written $\kappa \preceq n,\mu$ for short).
The first constraint comes from the fact that the row
canonical form of a matrix in $R^{n \times \mu}$ has at most
$n$ nonzero rows.  The second constraint comes from the fact
that $\row A$ is a submodule of $R^{\mu}$, for any $A \in R^{n
\times \mu}$.  Hence, we will assume that $\kappa \preceq n,
\mu$ in the rest of this paper.  As we will see, the set
$\calT_{\kappa}(R^{n \times \mu})$, together with the row
canonical forms in $\calT_{\kappa}(R^{n \times \mu})$, plays a
crucial role in our coding schemes.

We now enumerate the row canonical forms in
$\calT_{\kappa}(R^{n \times \mu})$.  We need the following
lemma.
\begin{lem}\label{lem:correspondence} There is a
one-to-one correspondence between row canonical forms in
$\calT_{\kappa}(R^{n \times \mu})$ and submodules of $R^{\mu}$
with shape $\kappa$.  \end{lem}

The proof is provided in Appendix~\ref{sec:appendix-proofs}.
By Lemma~\ref{lem:correspondence},  the number of row
canonical forms in $\calT_{\kappa}(R^{n \times \mu})$ is
$\submodule{\mu}{\kappa}$.  It is helpful to bound this number
as well as the logarithm of this number.  Combining
(\ref{eqn:ModuleCount}) and the fact that
\[
    q^{k(m - k)} \le \gauss{m}{k} \le4 q^{k(m - k)}
\]
(see, e.g., \cite[Lemma~4]{KK07}), we have
\begin{equation}\label{eq:num-module}
        q^{\sum_{i = 1}^\len \kappa_i(\mu_i - \kappa_i)} \le \submodule{\mu}{\kappa}
        \le 4^s q^{\sum_{i = 1}^\len \kappa_i (\mu_i - \kappa_i)}.
\end{equation}
Taking logarithms, we obtain
\begin{equation}\label{eq:log-module}
        \sum_{i = 1}^\len \kappa_i(\mu_i - \kappa_i) \le \log_q\submodule{\mu}{\kappa}
        \le \sum_{i = 1}^\len \kappa_i (\mu_i - \kappa_i)  + \len \log_q 4.
\end{equation}

\begin{exam}\label{ex:count-matrix}
Let $R = \ZZ_4$, and let $n = 2$, $\mu = (2, 3)$, $\kappa =
(1, 2)$.  Then by Lemma~\ref{lem:correspondence}, there are
$18$ row canonical forms in $\calT_{\kappa}(R^{n \times
\mu})$.  These $18$ row canonical forms can be classified into
$4$ categories based on the positions of their pivots:
\[
    \mat{1 & * & *\\0 & 2 & *} \ \mat{0 & 1 & *\\2 & 0 & *} \ \mat{1 & * & 0\\ 0 & 0 & 2} \
    \mat{* & 1 & 0\\ 0 & 0 & 2}  .
\]
The first category contains $8$ row canonical forms, namely,
\begin{align*}
    &\mat{1 & 0 & 0\\0 & 2 & 0} \ \mat{1 & 0 & 0\\0 & 2 & 2} \ \mat{1 & 0 & 2\\0 & 2 & 0} \
    \mat{1 & 0 & 2\\0 & 2 & 2}\\
    &\mat{1 & 1 & 0\\0 & 2 & 0} \ \mat{1 & 1 & 0\\0 & 2 & 2} \ \mat{1 & 1 & 2\\0 & 2 & 0} \
    \mat{1 & 1 & 2\\0 & 2 & 2}.
\end{align*}
The second category contains $4$ row canonical forms, namely,
\[
    \mat{0 & 1 & 0\\2 & 0 & 0} \ \mat{0 & 1 & 2\\2 & 0 & 0} \ \mat{0 & 1 & 0\\2 & 0 & 2} \
    \mat{0 & 1 & 2\\2 & 0 & 2}.
\]
The third category contains $4$ row canonical forms, namely,
\[
    \mat{1 & 0 & 0\\ 0 & 0 & 2} \ \mat{1 & 1 & 0\\ 0 & 0 & 2} \ \mat{1 & 2 & 0\\ 0 & 0 & 2} \
    \mat{1 & 3 & 0\\ 0 & 0 & 2}.
\]
The fourth category contains $2$ row canonical forms, namely,
\[
    \mat{0 & 1 & 0\\ 0 & 0 & 2} \ \mat{2 & 1 & 0\\ 0 & 0 & 2}.
\]
Clearly, the first category contains a significant portion of
all possible row canonical forms.
\end{exam}

Motivated by the above example, we introduce principal row
canonical forms that make up a significant portion of all
possible row canonical forms in $\calT_{\kappa}(R^{n \times
\mu})$.

A row canonical form in $\calT_{\kappa}(R^{n \times \mu})$ is
called \emph{principal} if its diagonal entries $d_1, d_2,
\ldots, d_r$ ($r = \min\{ n, m \}$) have the following form:
\begin{equation}\label{eq:principal-form}
d_1, \ldots, d_r \!=\! \underbrace{1, \ldots, 1}_{\kappa_1}, \underbrace{\pi, \ldots, \pi}_{\kappa_2 - \kappa_1}, \ldots, \underbrace{\pi^{\len - 1}, \ldots, \pi^{\len - 1}}_{\kappa_\len - \kappa_{\len - 1}},
\underbrace{0, \ldots, 0}_{r - \kappa_\len}.
\end{equation}
Clearly, the first category in Example~\ref{ex:count-matrix}
contains all principal row canonical forms for
$\calT_{\kappa}(\ZZ_4^{n \times \mu})$ with $n = 2$, $\mu =
(2, 3)$ and $\kappa = (1, 2)$.

\begin{prop}\label{prop:principal}
Every principal row canonical form $X \in \calT_{\kappa}(R^{n
\times \mu})$ can be constructed based on its $\pi$-adic
decomposition
\[
X = X_0 + \pi X_1 + \cdots + \pi^{\len - 1} X_{\len - 1},
\]
with each auxiliary matrix $X_i$ ($i = 0, \ldots, \len - 1$)
satisfying the following conditions:
\begin{enumerate}
\item $X_i[1 {:} \kappa_{i + 1}, 1 {:} \kappa_{i + 1}] = \diag(\underbrace{0, \ldots, 0}_{\kappa_i},
\underbrace{1, \ldots, 1}_{\kappa_{i+1} - \kappa_i})$,
\item $X_i[1{:} \kappa_{i+1}, \kappa_{i+1}+1 {:}
\mu_{i+1}]$ can be any matrix over $\calR(R, \pi)$, and
\item all other entries in $X_i$ are zero.
\end{enumerate}
\end{prop}

The proof is provided in Appendix~\ref{sec:appendix-proofs}.
The construction is illustrated in
Fig.~\ref{fig:construction}.  Clearly, this construction
provides a one-to-one mapping from sequences of $\sum_{i =
1}^{\len}  \kappa_i(\mu_i - \kappa_i)$ $q$-ary symbols to
principal row canonical forms in $\calT_{\kappa}(R^{n \times
\mu})$.  Note that the number of principal row canonical forms
in $\calT_{\kappa}(R^{n \times \mu})$ is $q^{\sum_{i =
1}^{\len}  \kappa_i(\mu_i - \kappa_i)}$, which is comparable
to the number of row canonical forms in $\calT_{\kappa}(R^{n
\times \mu})$ in total.

\begin{figure}[t]
  \centering
\ifCLASSOPTIONonecolumn
  \setlength{\tabcolsep}{1pt}
  \begin{tabular}{ccc}
     \includegraphics{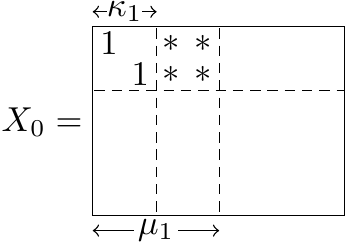} &
     \includegraphics{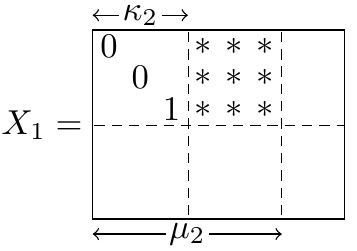} &
     \includegraphics{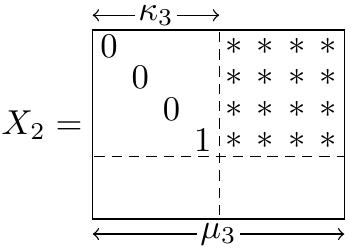}
   \end{tabular}
\else
  \scalebox{0.75}{\setlength{\tabcolsep}{1pt}\begin{tabular}{ccc}
     \includegraphics{fig4a} &
     \includegraphics{fig4b} &
     \includegraphics{fig4c}
   \end{tabular}}
\fi
  \caption{Illustration of the construction of principal row canonical forms
   for $\calT_{\kappa}(R^{n \times \mu})$ with
   $\len = 3$, $n=6$, $\mu=(4,6,8)$, and $\kappa = (2,3,4)$.}
  \label{fig:construction}
\end{figure}

\subsection{General Matrices in $\calT_\kappa(R^{n \times \mu})$}

Next, we count the number of matrices in
$R^{n \times \mu}$ of shape $\kappa$,
which is a central result in this section.
The proof is provided in Appendix~\ref{sec:appendix-proofs}.

\begin{thm}\label{thm:MatrixCount}
The size of $\calT_{\kappa}(R^{n \times \mu})$ is given by
\begin{equation}
|\calT_{\kappa}(R^{n \times \mu})| = |R^{n \times \kappa}|   \prod_{i = 0}^{\kappa_\len - 1} (1 - q^{i - n})
\submodule{\mu}{\kappa}.
\label{eqn:MatrixCount}
\end{equation}
\end{thm}

In particular, when the
chain length $\len = 1$, $R$ becomes $\FF_q$, and this
counting result becomes $\prod_{i = 0}^{\kappa_1- 1} (q^n -
q^{i}) \gauss{\mu_1}{\kappa_1}$, which is the number of $n
\times \mu_1$ matrices of rank $\kappa_1$.  We note that
Theorem~\ref{thm:MatrixCount} generalizes a theorem of
\cite{McDonald70} from square matrices to general matrices and
from Galois rings to finite chain rings.

Taking logarithms on both sides of (\ref{eqn:MatrixCount}), we have
\ifCLASSOPTIONonecolumn
\[
    \log_q |\calT_\kappa(R^{n \times \mu})| =\log_q \submodule{\mu}{\kappa} + \log_q |R^{n    \times \kappa}|
    + \log_q \prod_{i = 0}^{\kappa_\len - 1}(1 - q^{i - n}).
\]
\else 
\begin{multline*}
    \log_q |\calT_\kappa(R^{n \times \mu})| =\log_q \submodule{\mu}{\kappa} + \log_q |R^{n    \times \kappa}| \\
    + \log_q \prod_{i = 0}^{\kappa_\len - 1}(1 - q^{i - n}).
\end{multline*}
\fi
Combining this with (\ref{eq:log-ambient}) and (\ref{eq:log-module}), we obtain
\begin{multline}
        \sum_{i = 1}^\len \kappa_i(n + \mu_i - \kappa_i) + \log_q \prod_{i = 0}^{\kappa_\len - 1}(1 - q^{i - n}) \\
        \le \log_q |\calT_\kappa(R^{n \times \mu})| \le \\
        \sum_{i = 1}^\len \kappa_i(n + \mu_i - \kappa_i)
        + \log_q \prod_{i = 0}^{\kappa_\len - 1}(1 - q^{i - n}) + \len \log_q 4.
        \label{eq:log-matrix}
\end{multline}

\subsection{Notational Summary}

Table~\ref{table:notations}
summarizes the notation
that will be used extensively in the study of matrix channels.
Also listed are finite-field counterparts, which facilitates
comparisons of this work with \cite{SKK10}.

\begin{table}[hp]
\centering
\caption{Notational Summary}
\vspace{0.5em}
\small
{
\renewcommand{\arraystretch}{1.2}
\ifCLASSOPTIONonecolumn
\begin{tabular}{c|c|c}
\hline
 notation & meaning   & finite-field counterpart  \\
\hline \hline
$\mu$ & shape & rank \\ \hline
$R^{\mu}$ & $R$-module & vector space $\FF_q^m$ \\ \hline
$R^{n \times \mu}$ & set of matrices with rows from
$R^{\mu}$ & $\FF_q^{n \times m}$ \\ \hline
$\calT_{\kappa}(R^{n \times \mu})$ & set of matrices in
$R^{n \times \mu}$ with shape $\kappa$ & set of matrices
in $\FF_q^{n \times m}$ with rank $t$ \\ \hline
$\RCF(A)$ & row canonical form of $A$ & reduced row echelon form \\ \hline
\end{tabular}
\else
\begin{tabular}{c|c|c}
\hline
 notation & meaning   & counterpart  \\
 \hline \hline
 $\mu$ & shape & rank \\ \hline
 $R^{\mu}$ & $R$-module & vector space $\FF_q^m$ \\ \hline
  & \raisebox{-0.4em}{set of matrices with} & \\
  \raisebox{1em}{$R^{n \times \mu}$} & \raisebox{0.4em}{rows from
  $R^{\mu}$} & \raisebox{1em}{$\FF_q^{n \times m}$}\\ \hline
   & \raisebox{-0.4em}{set of matrices in} & \raisebox{-0.4em}{set of matrices in} \\
  \raisebox{1em}{$\calT_{\kappa}(R^{n \times \mu})$} & \raisebox{0.4em}{$R^{n \times \mu}$ with shape $\kappa$} & \raisebox{0.4em}{$\FF_q^{n \times m}$ with rank $t$}\\ \hline
  $\RCF(A)$ & row canonical form  & reduced row echelon form \\ \hline
\end{tabular}
\fi
}
\label{table:notations}
\end{table}

\section{Channel Decomposition}\label{sec:channel}

In this section, we introduce a channel decomposition technique that
converts a matrix channel over certain finite rings into a set of
independent parallel matrix channels over finite chain rings.  This
enables us to focus on matrix channels over finite chain rings, thereby
greatly facilitating our study of capacity results and coding schemes in
later sections.

As shown in our previous work \cite{FSK-submitted}, nested-lattice-based
PNC induces a message space of the form $\Omega = T / \langle d_1
\rangle \times \cdots \times T / \langle d_m \rangle$, where $T$ is a
PID and $d_m \mid \cdots \mid d_1$.  Let $R \triangleq T / \langle d_1
\rangle$. (Note that $R$ is a PIR, but not necessarily a finite chain ring.)
We can rewrite $\Omega$ as
\[
\Omega = R \times (d_1/d_2) R \times \cdots \times (d_1/d_m) R;
\]
this expression says that $\Omega$ can be viewed as a collection of
$m$-tuples (over $R$) whose $j$th component is a multiple of $d_1/d_j$.
\begin{exam}
Let $\Omega = \ZZ_{12} \times \ZZ_{6} \times \ZZ_{6} \times \ZZ_{2}$.
Then $\Omega$ can be expressed as $\ZZ_{12} \times 2 \ZZ_{12} \times 2
\ZZ_{12} \times 6 \ZZ_{12}$ via the following map:
\[
(a_1 + (12), a_2 + (6), a_3 + (6), a_4 + (2)) \to (a_1, 2 a_2, 2 a_3, 6 a_4),
\]
where $a_1 \in \{0, \ldots, 11\}$, $a_2, a_3 \in \{0, \ldots, 5 \}$, and
$a_4 \in \{ 0,1 \}$. Clearly, this map is one-to-one.
\end{exam}

With this expression, our matrix channel can be written as
\begin{equation}\label{eq:chain-ring-channel}
    Y = AX + BE
\end{equation}
where $X \in R^{n \times m}$ and $Y \in R^{N \times m}$ are the input
and output matrices whose rows are from $\Omega$, $E \in R^{t \times m}$
is the error matrix whose rows (also from $\Omega$) correspond to
additive (random) error packets.  The transfer matrices $A \in R^{N
\times n}$ and $B \in R^{N \times t}$ are random matrices with some
joint distribution, and $X$, $(A, B)$, $E$ are statistically
independent.  For simplicity of presentation, we sometimes write the
channel model as $Y = AX + Z$, where $Z = BE$ is called the noise
matrix.  Clearly, the channel model is an instance of the discrete
memoryless channel $(\calX, p_{Y|X}, \calY)$ with input alphabet $\calX
= R^{n \times m}$, output alphabet $\calY = R^{N \times m}$ and channel
transition probability $p_{Y|X}$.  The capacity of this channel is given
by
 \[
    C = \max_{p_X} I(X; Y)
 \]
where $p_X$ is the input distribution.

Next, we illustrate how to decompose the matrix channel.  To this end,
we first decompose the message space $\Omega$.  Since $T$ is a PID, $d_1
\in T$ can be factored as $d_1 = u_1 p_1^{t_{1,1}} \cdots
p_L^{t_{L,1}}$, where $u_1$ is a unit in $T$, $p_1, \ldots,
p_L$ are primes in $T$, and $t_{1,1}, \ldots, t_{L,1}$ are
positive integers.  Since  $d_m \mid \cdots \mid d_1$, we have $d_j =
u_j p_1^{t_{1,j}} \cdots p_L^{t_{L,j}}$ ($j = 2, \ldots, m$),
where $u_j$ is a unit, and $t_{1,j}, \ldots, t_{L,j}$ are
non-negative integers. Now, let
\[
\Omega_\ell \triangleq T / \langle p_\ell^{t_{\ell,1}} \rangle \times \cdots \times
T / \langle p_\ell^{t_{\ell,m}} \rangle, \ \ell = 1, \ldots, L.
\]
By the Chinese remainder theorem, we have $\Omega \cong \Omega_1 \times
\cdots \times \Omega_L$. This gives rise to a decomposition of
$\Omega$.
\begin{exam}
Let $\Omega = \ZZ_{12} \times \ZZ_{6} \times \ZZ_{6} \times \ZZ_{2}$. Then
\begin{align*}
\Omega &\cong (\ZZ_{4} \times \ZZ_3) \times( \ZZ_{2} \times \ZZ_3) \times (\ZZ_{2} \times \ZZ_3) \times \ZZ_{2} \\
&\cong (\underbrace{\ZZ_4 \times \ZZ_2 \times \ZZ_2 \times \ZZ_2}_{\Omega_1}) \times (\underbrace{\ZZ_3 \times \ZZ_3 \times \ZZ_3}_{\Omega_2}).
\end{align*}
\end{exam}

Note that $\Omega_\ell$ has an interesting interpretation:
$\Omega_\ell$ is a natural projection of $\Omega$ onto some finite chain
ring. Let $R_\ell \triangleq T / \langle p_\ell^{t_{\ell,1}} \rangle$ (which is a finite chain ring).
It is easy to check that $\Omega_\ell = R_\ell \times p_\ell^{(t_{\ell,1} - t_{\ell,2})} R_\ell \times
\cdots \times p_\ell^{(t_{\ell,1} - t_{\ell,m})} R_\ell$ and that
\[
\Omega_\ell = \{ (r_1, \ldots, r_m) \mbox{ mod } R_\ell \mid (r_1, \ldots, r_m) \in \Omega  \}.
\]

We are now ready to introduce the channel decomposition.  For any matrix
$X \in R^{n \times m}$,
let $X^{[\ell]} \triangleq X \mbox{ mod } R_\ell$, the
projection of every entry of $X$ onto $R_\ell$.
Applying this projection
to the matrix channel, we obtain $L$ sub-channels
\begin{equation}\label{eq:channel3}
Y^{[\ell]} = A^{[\ell]} X^{[\ell]} + Z^{[\ell]},
\end{equation}
for $\ell = 1, \ldots, L$, as illustrated in Fig.~\ref{fig:decomp}.
Clearly, each row of $X^{[\ell]}$ (or, $Y^{[\ell]}$, $Z^{[\ell]}$)
is from $\Omega_\ell$.

\begin{figure}[h]
\begin{center}
\scalebox{0.8}{
\includegraphics{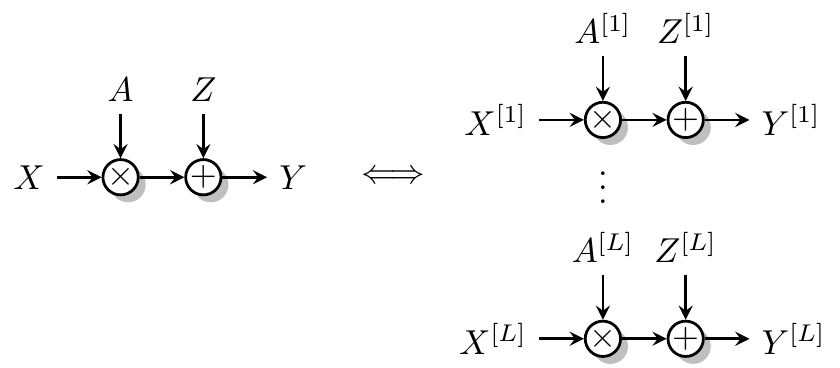}
}
\caption{An illustration of the channel decomposition.}
\label{fig:decomp}
\end{center}
\end{figure}

These sub-channels are, in general, correlated with each other.  Hence,
we have  $C \ge \sum_{\ell = 1}^L C_\ell$, where $C_\ell$ is the capacity of
sub-channel $\ell$. The equality is achieved for certain distributions of
$A$ and $Z$. One such distribution is provided in Theorem~\ref{thm:decomposition}.
We need a few definitions.
We say a matrix $A \in R^{n \times m}$ have \emph{rank} $t$, if for all $\ell$, $A^{[\ell]}$ has rank $t$. A matrix $A \in R^{n \times m}$ is \emph{full rank}
if $\rank A = \min\{n, m\}$.

\begin{thm}\label{thm:decomposition}
Suppose that the transfer matrix $A \in R^{N \times n}$ ($N \ge n$) is
uniform over all full-rank matrices and that the noise matrix $Z \in
R^{N \times m}$ is uniform over all rank-$t$ matrices (whose rows are
from $\Omega$).  Suppose that $A$ and $Z$ are independent of each other.
Then the channel decomposition induces $L$ independent sub-channels
\[
Y^{[\ell]} = A^{[\ell]} X^{[\ell]} + Z^{[\ell]}, \ \ell = 1, \ldots, L,
\]
where $A^{[\ell]} \in R_\ell^{N \times n}$ is uniform over full-rank matrices
(over $R_\ell$), $Z^{[\ell]}$ is uniform over rank-$t$ matrices whose rows are
from $\Omega_\ell$, and $A^{[\ell]}$ is independent of $Z^{[\ell]}$.  Clearly,
these sub-channels form a product discrete memoryless channel (DMC).  In
particular, the capacity of this product DMC is $C = \sum_{\ell = 1}^L C_\ell$.
\end{thm}
\begin{proof}
Note that $A$ is full rank over $R$, if and only if each $A^{[\ell]}$ is
full rank over $R_\ell$.  Hence, the number of full-rank matrices in $R^{N
\times n}$ is equal to the product of the number of full-rank matrices
in $R_\ell^{N \times n}$ ($\ell = 1, \ldots, L$).  In particular,
it follows that when $A$ is uniform
over full-rank matrices, each $A^{[\ell]}$ is also uniform over full-rank
matrices and independent of each other.  Similarly, each $Z^{[\ell]}$ is
uniform over rank-$t$ matrices and independent of each other.  Since
$A^{[\ell]}$ and $Z^{[\ell]}$ are projections of $A$ and $Z$, respectively,
$A^{[\ell]}$ and $Z^{[\ell]}$ are independent.  Therefore, the sub-channels
$Y^{[\ell]} = A^{[\ell]} X^{[\ell]} + Z^{[\ell]}$ are independent of each other. In
particular, $C = \sum_{\ell = 1}^L C_\ell$.
\end{proof}

Theorem~\ref{thm:decomposition} says that
when $A$ and $Z$ follow certain distributions, the channel decomposition
incurs no loss of information.  Hence, in this case, it suffices to
study each sub-channel independently.

Next, we comment on the assumptions in Theorem~\ref{thm:decomposition}.
First, as we will soon see in later sections, these assumptions allow us to
derive clean capacity results and simple coding schemes, based on which
more general distributions can be studied (see Section~\ref{sec:extension}).

Second, we note that the full-rank assumption on $A$
and the rank-$t$ assumption on $Z$ are reasonable,
when the system size is large.
To see this, observe that
the portion of full-rank matrices in $R^{N \times n}$
is lower-bounded by
\[
1 - \sum_{\ell = 1}^L \frac{n}{|p_\ell|^{2(1 + N - n)}}.
\]
Clearly, this lower bound tends to $1$ as $n$ and $N$ grow.
For example, if we set $n = 100$, $N = 110$, and choose
$R = \ZZ_2[i] = \ZZ[i]/\langle (1+i)^2 \rangle$, then the lower bound is around $0.999976$.
Using the same argument, we can show that rank-$t$ matrices make up a significant portion
of all possible noise matrices  $Z = BE$
for large $t$, $m$, and $N$.

Third, we note that the uniformness assumptions
on $A$ and $Z$ provide us with
``worst-case'' scenarios, which will be elaborated in
Section~\ref{sec:extension}.

Without loss of generality, we will focus on the case $L = 1$, and
so $R$ is a finite chain ring for the remainder of the paper.
Suppose that $R$ be a $(q, \len)$ chain ring. Let $\mu$ be the shape of $\Omega$. Then, we can write
$X \in R^{n \times \mu}$ and $Y, Z \in R^{N \times
\mu}$.  That is, we may think of the rows of $X$, $Y$ and
$Z$ as packets over the ambient space $R^\mu$. (To support this
ambient space, the length of a packet, denoted by $m$, is equal to
$\mu_\len$.)

In many situations, it is useful to understand the capacity scaling as
the system size and packet length grow.  For that reason, we introduce a
notion of asymptotic capacity
\[
  \bar{C} = \lim_{m \to \infty} \frac{1}{n|\mu|} C  = \lim_{m \to \infty} \frac{1}{\bar{n}|\bar{\mu}|m^2} C,
\]
where we assume that $\bar{n} = n/m$ and $\bar{\mu} =
(\bar{\mu}_1,\ldots,\bar{\mu}_\len) = \mu/m$ are fixed.  Here,
logarithms are taken to the base $q$, so that the capacity $C$ is
given in $q$-ary units per channel use and that $\bar{C}$ is
normalized such that $\bar{C} = 1$ if the channel is noiseless (i.e., $A
= I$ and $Z = 0$).

\section{The Multiplicative Matrix Channel}\label{sec:mmc}

As a first special case, following \cite{SKK10}, we consider the
\emph{multiplicative matrix channel (MMC)} defined by the law
\[
  Y = AX,
\]
where $A \in R^{N \times n}$ is uniform over all
full-column-rank matrices and independent from $X \in R^{n \times \mu}$.
This model is a special case of the channel model
(\ref{eq:channel3}) with $Z = 0$.

\subsection{Capacity}

The capacity of the  MMC can be obtained by investigating the
channel transition probabilities.
Since full-column-rank matrices preserve the row span, we have
$\row X = \row Y$. It follows that the channel transition probability $p_{Y|X}(Y | X) > 0$ if and only if $\row X = \row Y$.
Moreover, we have the following lemma:
\begin{lem}\label{lem:channel-transition}
The channel transition probabilities satisfy the following two
properties.
\begin{enumerate}
\item $p_{Y|X}(Y_1 | X) = p_{Y|X}(Y_2 | X) > 0$, if $\row X =
\row Y_1 = \row Y_2$.
\item $p_{Y|X}(Y | X_1) = p_{Y|X}(Y | X_2) > 0$, if $\row X_1 =
\row X_2 = \row Y$.
\end{enumerate}
\end{lem}
\begin{proof}
Since $\row Y_1 = \row Y_2$, there exists some invertible matrix $P$
such that $P Y_1 = Y_2$.
Let $\calA_j = \{ A \in \calT_n(R^{N \times n}) \mid A X = Y_j  \}$ be the set of transfer matrices
such that $A X = Y_j$. Then $\calA_1$ and $\calA_2$ have the
same size (i.e., $|\calA_1| = |\calA_2|$),
because $A \in \calA_1$ if and only if $PA \in \calA_2$.
Hence, we have $p_{Y|X}(Y_1 | X) = p_{Y|X}(Y_2 | X)$.
In particular, when $\row X = \row Y_1$, the set $\calA_1$ is non-empty, and so $p_{Y|X}(Y_1 | X) > 0$.
This proves Part 1). Similarly, we can prove Part 2).
\end{proof}

Lemma~\ref{lem:channel-transition} characterizes the structure of the
channel transition probabilities, based on which one can show that
the capacity only depends on the number of all possible submodules
generated by $X$.

\begin{thm}\label{thm:cap-mmc}
The capacity of the MMC, in $q$-ary symbols per channel use,
is given by
  \[
    C_\text{MMC} = \log_q \sum_{\lambda \preceq n, \mu}
    \submodule{\mu}{\lambda}.
  \]
A capacity-achieving code $\calC \subseteq R^{n \times \mu}$
consists of all possible row canonical forms in $R^{n \times
\mu}$.
\end{thm}

Theorem~\ref{thm:cap-mmc} suggests that information
should be encoded in the choice of submodules. That is,
``transmission via submodules'' is optimal here. This naturally
generalizes the ``transmission via subspaces'' strategy in
\cite{KK07}.


\begin{cor}\label{cor:cap-mmc}
The capacity $C_\text{MMC}$ is bounded by
  \begin{equation}\label{eq:mmc-bound}
    \sum_{i = 1}^\len \kappa_i (\mu_i - \kappa_i) \le
    C_\text{MMC} \le
    \sum_{i = 1}^\len \kappa_i (\mu_i - \kappa_i)  + \log_q 4^\len \binom{n + \len}{\len}
  \end{equation}
  where $\kappa_i = \min\{ n, \lfloor \mu_i/2 \rfloor \}$ for all $i$.
\end{cor}
\begin{proof}
First, since $\kappa = (\kappa_1, \ldots, \kappa_s) \preceq n,
\mu$, we have
\begin{align*}
C_\text{MMC} &= \log_q \sum_{\lambda \preceq n, \mu} \submodule{\mu}{\lambda} \\
&\ge \log_q \submodule{\mu}{\kappa} \\
&\ge  \sum_{i = 1}^\len \kappa_i (\mu_i - \kappa_i),
\end{align*}
where the second inequality follows from
(\ref{eq:log-module}).

Second, we have
\begin{align*}
C_\text{MMC} &= \log_q \sum_{\lambda \preceq n, \mu} \submodule{\mu}{\lambda} \\
  &\le \log_q \sum_{\lambda \preceq n,\mu} 4^\len q^{\sum_i \lambda_i(\mu_i - \lambda_i)} \\
  &\le \log_q \sum_{\lambda \preceq n,\mu} 4^\len q^{\sum_i \kappa_i(\mu_i - \kappa_i)} \\
  &\le \log_q  4^\len \binom{ n + \len}{\len} q^{\sum_i \kappa_i(\mu_i - \kappa_i)} \\
  &= \sum_{i = 1}^\len \kappa_i (\mu_i - \kappa_i)  +  \log_q 4^{\len} \binom{n + \len }{\len}.
\end{align*}
where the first inequality follows from (\ref{eq:log-module}),
the second inequality follows from the fact that $\kappa$
maximizes the quantity $\sum_i \lambda_i(\mu_i - \lambda_i)$
subject to the constraint $\lambda \preceq n, \mu$, and the
third inequality follows from the fact that the number of
shapes satisfying $\lambda \preceq n, \mu$ is upper-bounded by
$\binom{ n + \len}{\len}$.
\end{proof}

We next turn to the asymptotic capacity of the MMC.

\begin{thm}\label{thm:asymp-mmc}
The asymptotic capacity $\bar{C}_\text{MMC}$ is given by
\begin{equation}\label{eq:mmc-asymp}
   \bar{C}_\text{MMC} = \frac{\sum_{i = 1}^{\len}
    \bar{\kappa}_i (\bar{\mu}_i - \bar{\kappa}_i) }
    {\bar{n} |\bar{\mu}|},
\end{equation}
where $\bar{\kappa} = \kappa/m$ with $\kappa_i = \min \{ n,
\lfloor \mu_i / 2 \rfloor \}$ for all $i$.
\end{thm}
\begin{proof}
This follows from Corollary~\ref{cor:cap-mmc} and the fact
that $\frac{1}{m^2} \log_q 4^{\len} \binom{n + \len }{\len}
\to 0$, as $m \to \infty$.
\end{proof}

Theorem~\ref{thm:asymp-mmc} implies that the shape $\kappa$
given by $\kappa_i = \min \{ n, \lfloor \mu_i / 2 \rfloor \}$
($1 \leq i \le \len$) is ``typical'' among the shapes of all
possible row canonical forms in $R^{n \times \mu}$.  In other
words, the row canonical forms of shape $\kappa$ make up a
significant portion of all possible row canonical forms.
Hence, the transmitter may encode information in the choice of
row canonical forms of shape $\kappa$ instead of all row
canonical forms.

\subsection{A Simple Coding Scheme}

In this section, we present a simple coding scheme that
achieves the asymptotic capacity in
Theorem~\ref{thm:asymp-mmc}.  The key idea is to make the
codebook the set of all principal row canonical forms for
$\calT_\kappa(R^{n \times \mu})$.  In other words, we employ
two ``reductions'' in the code construction.  First, we move
from all row canonical forms in $R^{n \times \mu}$ to all row
canonical forms in $\calT_\kappa(R^{n \times \mu})$, as
suggested by Theorem~\ref{thm:asymp-mmc}.  Then, we move from
all row canonical forms in $\calT_\kappa(R^{n \times \mu})$ to
all principal row canonical forms in $\calT_\kappa(R^{n \times
\mu})$.  With these two reductions, our coding scheme not only
achieves the asymptotic capacity, but also admits fast
encoding and decoding.

\subsubsection{Encoding}

The input matrix $X$ is chosen from the set of principal row
canonical forms for $\calT_{\kappa}(R^{n \times \mu})$ by
using the construction presented in
Section~\ref{sec:construction}.  Clearly, the encoding rate of
the scheme is $R_\text{MMC} = \sum_{i = 1}^{\len}
\kappa_i(\mu_i - \kappa_i)$.

\subsubsection{Decoding}

Upon receiving $Y = AX$, the decoder simply computes the row
canonical form of $Y$.  The decoding is always correct by the
uniqueness of the row canonical form. By comparing the
encoding rate with the asymptotic capacity, we have the
following theorem.
\begin{thm}
The coding scheme described above achieves the asymptotic
capacity (\ref{eq:mmc-asymp}).
\end{thm}

\section{The Additive Matrix Channel}\label{sec:amc}

In this section, we consider the \emph{additive matrix channel
(AMC)} defined by the law
\[
  Y = X + Z,
\]
where $Z$ is uniform over $\calT_\tau(R^{n \times \mu})$ and
independent from $X$.  This model is a special case of the
channel model (\ref{eq:channel3}) with $A = I$.

\subsection{Capacity}
\begin{thm}\label{thm:cap-amc}
The capacity of the AMC, in $q$-ary symbols per channel use,
is given by
\[
    C_\text{AMC} = \log_q|R^{n \times \mu}| - \log_q |\calT_\tau(R^{n \times \mu})|,
\]
achieved by the uniform input distribution.
\end{thm}
\begin{proof}
The AMC is an example of a symmetric discrete memoryless
channel, whose capacity is achieved by the uniform input
distribution.  Note that when $X$ is uniform over $R^{n \times
\mu}$, so is $Y$.  Thus, we have
\[
    C_\text{AMC} = H(Y) - H(Z) =  \log_q|R^{n \times \mu}| - \log_q |\calT_\tau(R^{n \times \mu})|.
\]
\end{proof}
\begin{cor}\label{cor:cap-amc}
The capacity $C_\text{AMC}$ is bounded by
\ifCLASSOPTIONonecolumn
\[
        \sum_{i = 1}^\len (n - \tau_i)(\mu_i - \tau_i) - \log_q 4^\len \prod_{i = 0}^{\tau_\len - 1}(1 - q^{i - n})
        < C_\text{AMC} <
        \sum_{i = 1}^\len (n - \tau_i)(\mu_i - \tau_i) - \log_q \prod_{i = 0}^{\tau_\len - 1}(1 - q^{i - n}).
\]
\else
\begin{multline}
        \sum_{i = 1}^\len (n - \tau_i)(\mu_i - \tau_i) - \log_q 4^\len \prod_{i = 0}^{\tau_\len - 1}(1 - q^{i - n}) \\
        < C_\text{AMC} < \\
        \sum_{i = 1}^\len (n - \tau_i)(\mu_i - \tau_i) - \log_q \prod_{i = 0}^{\tau_\len - 1}(1 - q^{i - n}).
\end{multline}
\fi
\end{cor}
\begin{proof}
It follows immediately from Theorem~\ref{thm:cap-amc}
and (\ref{eq:log-ambient}), (\ref{eq:log-matrix}).
\end{proof}

We next turn to the asymptotic behavior of the AMC.

\begin{thm}\label{thm:amc-asymp}
The asymptotic capacity $\bar{C}_\text{AMC}$ is given by
  \begin{equation}\label{eq:amc-asymp}
    \bar{C}_\text{AMC} = \frac{\sum_{i=1}^{\len} (\bar{n} - \bar{\tau}_i)(\bar{\mu}_i - \bar{\tau}_i)}
    {\bar{n} |\bar{\mu}|}.
  \end{equation}
\end{thm}
\begin{proof}
It follows from Corollary~\ref{cor:cap-amc} and the fact that
\[
\frac{1}{m^2} \log_q 4^\len \prod_{i = 0}^{\tau_\len - 1}(1 - q^{i - n})
 \to 0, \text{ as } m \to \infty.
\]
\end{proof}

\subsection{Coding Scheme}
We focus on a special case when $\tau = t$, and present a
coding scheme based on the idea of error-trapping in
\cite{SKK10}.  This scheme achieves the asymptotic capacity
for this special case.

\subsubsection{Encoding}
Set $v \ge t$. The input matrix $X$ is constructed as
\[
       X = \mat{0 & 0 \\
        0 & U},
\]
where the size of $U$ is $(n - v) \times (m - v)$, and the
sizes of other zero matrices are chosen to make $X$ an $n
\times m$ matrix.  Here, $U$ is chosen from the set $R^{(n -
v) \times (\mu - v)}$ by using the construction in
Section~\ref{sec:matrix-constraint} (as illustrated in
Fig.~\ref{fig:amc-scheme}). Clearly, the encoding rate of the
scheme is $R_\text{AMC} = \sum_{i=1}^{\len} (n - v)(\mu_i -
v)$.

\begin{figure}[t]
  \centering
\ifCLASSOPTIONonecolumn
  \setlength{\tabcolsep}{1pt}
  \begin{tabular}{ccc}
     \includegraphics{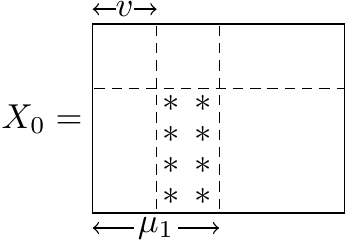} &
     \includegraphics{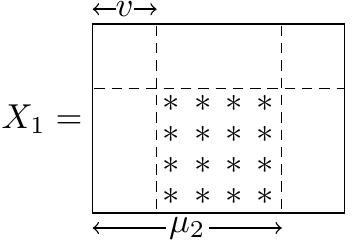} &
     \includegraphics{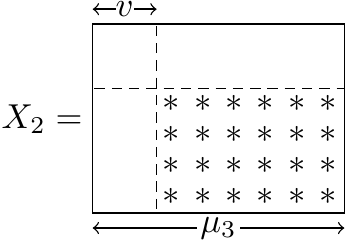}
   \end{tabular}
\else
  \scalebox{0.75}{\setlength{\tabcolsep}{1pt}\begin{tabular}{ccc}
     \includegraphics{fig6a} &
     \includegraphics{fig6b} &
     \includegraphics{fig6c}
   \end{tabular}}
\fi
\caption{Illustration of the AMC encoding scheme for $\len =
3$, $n=6$, $\mu=(4,6,8)$, and $v = 2$.}
  \label{fig:amc-scheme}
\end{figure}

\subsubsection{Decoding}

Following \cite{SKK10}, we write the noise matrix $Z$ as
\[
Z = BE = \mat{B_1 \\ B_2} \mat{E_1 & E_2},
\]
where $B_1 \in R^{v \times t}$, $B_2 \in R^{(n - v)\times t}$,
$E_1 \in R^{t \times v}$ and $E_2 \in R^{t \times (m - v)}$.
The received matrix $Y$ is then given by
\[
    Y = X + Z = \mat{B_1 E_1 & B_1 E_2 \\ B_2 E_1 & U + B_2 E_2}.
\]

Similar to \cite{SKK10}, we define that the error trapping is
successful if $\shape B_1 E_1 = t$. Assume that this is the
case.  Then by Proposition~\ref{prop:shape-properties}.3, we
have $\shape B_1 = \shape E_1 = t$. Consider the submatrix
consisting of the first $v$ columns of $Y$. Since $\shape B_1
E_1 = t$, the rows of $B_2 E_1$ are completely spanned by the
rows of $B_1 E_1$.  That is, $\row B_2 E_1 \subseteq \row B_1
E_1$.  Thus, there exists some matrix $\bar{T}$ such that $B_2
E_1 = \bar{T} B_1 E_1$.  Since $E_1$ is full row rank, by
Lemma~\ref{lem:column-rank}, $B_2 E_1 = \bar{T} B_1 E_1$ implies
$B_2 = \bar{T} B_1$. It follows that
\[
    T \mat{B_1 \\ B_2} = \mat{B_1 \\ 0},
   \mbox{ where } T = \mat{I & 0 \\ -\bar{T} & I}.
\]
Note also that $T X = X$. Thus,
\[
    TY = TX + TZ = \mat{B_1 E_1 & B_1 E_2 \\ 0 & U},
\]
from which the data matrix $U$ is readily obtained.

The decoding is summarized as follows.  The decoder observes
$B_1 E_1$, $B_1 E_2$, and $B_2 E_1$ thanks to the error traps.
The decoder then checks the condition $\shape B_1 E_1 = t$.
If the condition does not hold, the decoder declares a
failure.  Otherwise, the decoder finds a matrix $\bar{T}$ such
that $B_2 E_1 = \bar{T} B_1 E_1$ (which means $B_2 = \bar{T}
B_1$).  Since $B_2 = \bar{T} B_1$, the decoder can recover
$B_2 E_2$ by using the relation $B_2 E_2 = \bar{T} B_1E_2$.
Clearly, the error probability of the scheme is zero.  The
failure probability of the scheme is
\[
P_f =\Pr[\shape B_1 E_1~\ne~t].
\]

\begin{lem}\label{lem:amc-prob}
The failure probability $P_f$ of the above scheme is
upper-bounded by $P_f < \frac{2 t}{q^{1 + v - t}}$.
\end{lem}
\begin{proof}
If $B_1$ and $E_1$ are full rank, then $\shape B_1
E_1 = t$.  Hence, by the union bound, the failure probability
\[
P_f \le \Pr[E_1 \mbox{ is not full rank}]
+ \Pr[B_1 \mbox{ is not full rank}].
\]
Now consider the probability that $E_1$ is full rank.  Recall
that $E \in R^{t \times \mu}$ is a full-rank matrix chosen
uniformly at random.  An equivalent way of generating $E$ is
to first generate the entries of a matrix $E' \in R^{t \times
\mu}$ uniformly at random, and then discard $E'$ if it is not
full rank. This suggests that
\begin{align*}
\Pr[E_1 \mbox{ is full rank}] &= \Pr[E_1'  \mbox{ is full rank}
\mid E'  \mbox{ is full rank}] \\
&> \Pr[E_1'  \mbox{ is full rank}],
\end{align*}
where $E_1'$ consists of the first $v$ columns of $E'$.  Thus,
\begin{align*}
\Pr[E_1 \mbox{ is full rank}]
        &> |\calT_t(R^{t \times v})|/|R^{t \times v}|\\
        &= q^{\len t v} \prod_{i = 0}^{t - 1}
           (1 - q^{i - v}) / q^{\len t v}\\
        &= \prod_{i = 0}^{t - 1} (1 - q^{i - v})\\
        &> 1 - \frac{t}{q^{1 + v - t}}.
\end{align*}
Similarly, we can show that
\[
\Pr[B_1 \mbox{ is full rank}] > 1 - \frac{t}{q^{1 + v - t}}.
\]
Therefore, the failure probability
$P_f < \frac{2 t}{q^{1 + v - t}}$.
\end{proof}

Recall that the encoding rate of the scheme is $R_\text{AMC} =
\sum_{i=1}^{\len} (n - v) (\mu_i - v)$.  Thus, if we set $v$
such that
\[
v - t \to \infty, \mbox{ and } \frac{v - t}{m} \to 0,
\]
as $m \to \infty$, then we have $P_f \to 0$ and
$\bar{R}_\text{AMC} = \frac{R_\text{AMC}}{n |\mu |} \to
\bar{C}_\text{AMC}$.  Therefore, we have the following
theorem.

\begin{thm}
The coding scheme described above can achieve the capacity
expression (\ref{eq:amc-asymp}) for the special case when
$\tau = t$.
\end{thm}

{\bf Remark:} The general case can also be handled by combining the above scheme
with the successive cancellation technique.

\section{The Multiplicative-Additive Matrix Channel}\label{sec:ammc}

In this section, we consider the \emph{multiplicative-additive
matrix channel (MAMC)} defined by the law
\[
  Y = AX + Z,
\]
where $A \in \calT_n(R^{N \times n})$ and $Z \in \calT_\tau(R^{N \times
\mu})$ are uniformly distributed and independent from any
other variables.

\subsection{Capacity Bounds}

Since $A$ is uniform over $\calT_n(R^{N \times n})$, $A$ is statistically equivalent to $P \mat{0 \\ I_n}$,
where $P \in R^{N \times N}$ is uniform over $\GL_N(R)$,
$I_n \in R^{n \times n}$ is an identity matrix, and
$0 \in R^{(N-n) \times n}$ is a zero matrix. Hence, we have
\[
Y = P \mat{0 \\ I_n} X + Z = P \mat{0 \\ X} + Z = P\left( \mat{0 \\ X} + W \right),
\]
where $W = P^{-1}Z$ is uniform over $\calT_\tau(R^{N \times
\mu})$ and independent of $X$.

\begin{thm}\label{thm:cap-ammc}
The capacity of the MAMC, in $q$-ary symbols per channel
use, is upper-bounded by
\ifCLASSOPTIONonecolumn
\begin{equation}\label{eq:cap-ammc}
    C_\text{AMMC} \leq  \log_q \sum_{\lambda \preceq N,n + \tau, \mu}
        \submodule{\mu}{\lambda} - \log_q |\calT_\tau(R^{N \times \mu})|
   + \log_q \sum_{\tau' \preceq \tau} |\calT_{\tau'}(R^{N \times \min\{n+\tau_s, N\}})|.
\end{equation}
\else
\begin{multline}\label{eq:cap-ammc}
    C_\text{AMMC} \leq  \log_q \sum_{\lambda \preceq N,n + \tau, \mu}
        \submodule{\mu}{\lambda} - \log_q |\calT_\tau(R^{N \times \mu})|\\
   + \log_q \sum_{\tau' \preceq \tau} |\calT_{\tau'}(R^{N \times \min\{n+\tau_s, N\}})|.
\end{multline}
\fi
\end{thm}
\begin{proof}
Let $U = \mat{0 \\ X} + W$.
Then $Y = P U$, and $X$, $U$, $Y$ form a Markov chain. Hence,
$I(X; Y | U) = 0$. Using the chain rules, we have
\begin{align*}
I(X;Y) &= I(U; Y) - I(U; Y | X) + \underbrace{I(X; Y|U)}_{= 0}\\
&= I(U; Y) - H(U|X) + H(U | X, Y)\\
&= I(U; Y) - H(W) + H(W | X, Y)\\
&= I(U; Y) - \log_q |\calT_\tau(R^{N \times \mu})|
+ H(W | X, Y)
\end{align*}

Next, we upper bound the terms $I(U; Y)$ and $H(W|X,Y)$.
Since $\shape U \preceq N, n + \tau$, the row span $\row U$
has at most $\sum_{\lambda \preceq N,n + \tau, \mu}\submodule{\mu}{\lambda}$ choices. Hence,
$I(U; Y) \le \log_q \sum_{\lambda \preceq N,n + \tau, \mu}
        \submodule{\mu}{\lambda}$.

Let $\kappa =
\shape Y$. Let $S$ be the Smith normal form of $Y$.  Then $S$
contains $\kappa_\len$ nonzero diagonal entries. Thus, $Y$ can
be expressed as
\[
  Y = \mat{P_1 & P_2} \mat{S_{11} & 0 \\ 0 & 0} \mat{Q_1 \\ Q_2} = P_1 S_{11} Q_1,
\]
where $P_1 \in R^{N \times \kappa_\len}$, $Q_1 \in R^{\kappa_\len \times m}$,
and $S_{11} \in R^{\kappa_\len \times \kappa_\len}$.

Note that
  \[
  \mat{0 \\ X} + W = P^{-1} Y = P^{*} Q_1,
  \]
where $P^{*} = P^{-1} P_1 S_{11}$.  Since $Q_1$ consists of
the first $\kappa_\len$ rows of an invertible matrix $Q$,
$Q_1$ is a full-rank matrix. In particular, $Q_1$ contains an
invertible $\kappa_\len \times \kappa_\len$ submatrix.  By
reordering columns if necessary, we can assume that the left
$\kappa_\len \times \kappa_\len$ submatrix of $Q_1$ is
invertible.  Write $Q_1 = \mat{Q_{11} & Q_{12}}$, $X =
\mat{X_1 & X_2}$ and $W = \mat{W_1 & W_2}$, where $Q_{11}$,
$X_1$, and $W_1$ have $\kappa_\len$ columns.  We have
\[
\mat{0 & 0 \\ X_1 & X_2} + \mat{W_1 & W_2} = \mat{P^* Q_{11} & P^*
Q_{12}}.
\]
It follows that
\[
  P^{*} = \left(\mat{0 \\ X_1} + W_1 \right)Q_{11}^{-1} \mbox{ and }  W_2 = P^{*} Q_{12} - \mat{0 \\ X_2}.
\]
This suggests that $W_2$ can be computed from $W_1$ if $X$ and
$Y$ are known.  Thus,
\[
  H(W|X, Y) = H(W_1 | X, Y) \le H(W_1 | \shape Y).
\]
Since $W_1$ is an $N \times \kappa_\len$ matrix with $\shape
W_1 \preceq \tau$, we have
\[
  H(W_1 | \shape Y = \kappa) \le \log_q \sum_{\tau' \preceq \tau} |\calT_{\tau'}(R^{N \times \kappa_\len})|,
\]
which is maximized when $\kappa_\len = \min\{N, n + \tau_\len \}$. Hence,
\[
  H(W_1 | \shape Y) \le \log_q \sum_{\tau' \preceq \tau} |\calT_{\tau'}(R^{N \times \min\{n + \tau_\len, N \}})|.
\]
So, $H(W|X,Y) \le  \log_q \sum_{\tau' \preceq
\tau} |\calT_{\tau'}(R^{N \times \min\{n + \tau_\len, N \}})|$, which completes the
proof.
\end{proof}

\begin{cor}\label{cor:cap-ammc}
The capacity $C_\text{MAMC}$ is upper-bounded by
\ifCLASSOPTIONonecolumn
\[
    C_\text{MAMC} \le \sum_{i = 1}^{\len} (\mu_i - \xi_i) \xi_i +
    \sum_{i = 1}^{\len} (\min\{n+\tau_s, N \} - \mu_i) \tau_i + 2\len \log_q 4
     + \log_q\mbox{$\binom{N+\len}{\len}$}
    + \log_q \mbox{$\binom{\tau_\len+\len}{\len}$} -
       \log_q \prod_{i = 0}^{\tau_\len - 1}(1 - q^{i-N}),
\]
\else 
\begin{multline*}
    C_\text{MAMC} \le \sum_{i = 1}^{\len} (\mu_i - \xi_i) \xi_i +
    \sum_{i = 1}^{\len} (\min\{n+\tau_s, N \} - \mu_i) \tau_i \\
     + \log_q\mbox{$\binom{N+\len}{\len}$}
    + \log_q \mbox{$\binom{\tau_\len+\len}{\len}$} -
       \log_q \prod_{i = 0}^{\tau_\len - 1}(1 - q^{i-N}) + 2\len \log_q 4,
\end{multline*}
\fi
where $\xi_i = \min\{ N, n + \tau_i, \lfloor \mu_i/2 \rfloor \}$ for all
$i$.  In particular, when $\mu \succeq 2N$ and $\tau = t$, the upper bound
reduces to
\ifCLASSOPTIONonecolumn
\[
    C_\text{MAMC} \le \sum_{i = 1}^{\len} (\min\{n+t, N \} - t) (\mu_i - \min\{n+t, N \})
    + 2\len \log_q 4 \\ + \log_q \mbox{$\binom{N+\len}{\len}$}
    + \log_q \mbox{$\binom{t+\len}{\len}$} - \log_q \prod_{i = 0}^{t - 1}(1 - q^{i-N}).
\]
\else 
\begin{multline*}
    C_\text{MAMC} \le \sum_{i = 1}^{\len} (\min\{n+t, N \} - t) (\mu_i - \min\{n+t, N \}) \\
    + 2\len \log_q 4 \\ + \log_q \mbox{$\binom{N+\len}{\len}$}
    + \log_q \mbox{$\binom{t+\len}{\len}$} - \log_q \prod_{i = 0}^{t - 1}(1 - q^{i-N}).
\end{multline*}
\fi
\end{cor}
\begin{proof}
By (\ref{eq:mmc-bound}), we have
\[
    \log_q \sum_{\lambda \preceq N,n + \tau, \mu}
            \submodule{\mu}{\lambda} \le \sum_{i = 1}^\len (\mu_i - \xi_i) \xi_i + \len \log_q 4 +
    \log_q \binom{N + \len }{\len}.
\]

By (\ref{eq:log-matrix}), we have
\ifCLASSOPTIONonecolumn
\[
    - \log_q |\calT_\tau(R^{N \times \mu})| \le - \sum_{i = 1}^\len (N + \mu_i - \tau_i) \tau_i
    - \log_q \prod_{i = 0}^{\tau_\len - 1} (1 - q^{i - N}).
\]
\else 
\begin{multline*}
    - \log_q |\calT_\tau(R^{n \times \mu})| \le - \sum_{i = 1}^\len (n + \mu_i - \tau_i) \tau_i \\
    - \log_q \prod_{i = 0}^{\tau_\len - 1} (1 - q^{i - n}).
\end{multline*}
\fi
Note that
\ifCLASSOPTIONonecolumn
\[
    |\calT_{\tau'}(R^{N \times \min\{n + \tau_\len, N \}}| \le  |R^{N \times \tau'}| \submodule{\min\{n + \tau_\len, N\}}{\tau'} \le
    4^{\len} q^{\sum_{i=1}^\len (N + \min\{n + \tau_\len, N\} - \tau'_i)\tau'_i},
\]
\else 
\begin{multline*}
    |\calT_{\tau'}(R^{N \times \min\{n + \tau_\len, N \}}| \le  |R^{N \times \tau'}| \submodule{\min\{n + \tau_\len, N\}}{\tau'}\\
    \le
        4^{\len} q^{\sum_{i=1}^\len (N + \min\{n + \tau_\len, N\} - \tau'_i)\tau'_i},
\end{multline*}
\fi
where the first inequality comes from (\ref{eqn:MatrixCount}),
and the second inequality comes from (\ref{eq:log-ambient})
and (\ref{eq:log-module}).  Hence,
\ifCLASSOPTIONonecolumn
\begin{align*}
\sum_{\tau' \preceq \tau} |\calT_{\tau'}(R^{N \times \min\{n + \tau_\len, N \}}| &\le \sum_{\tau' \preceq \tau}
4^{\len} q^{\sum_{i=1}^\len (N + \min\{n + \tau_\len, N\} - \tau'_i)\tau'_i}\\
&\le {\binom{\tau_\len+\len}{\len}} 4^\len
q^{\sum_{i=1}^\len (N + \min\{n + \tau_\len, N\} - \tau_i)\tau_i}
\end{align*}
\else
\begin{align*}
&\quad \sum_{\tau' \preceq \tau} |\calT_{\tau'}(R^{N \times \min\{n + \tau_\len, N \}}| \\
&\le \sum_{\tau' \preceq \tau}
4^{\len} q^{\sum_{i=1}^\len (N + \min\{n + \tau_\len, N\} - \tau'_i)\tau'_i}\\
&\le {\binom{\tau_\len+\len}{\len}} 4^\len
q^{\sum_{i=1}^\len (N + \min\{n + \tau_\len, N\} - \tau_i)\tau_i}
\end{align*}
\fi
where the second inequality comes from the fact that $\tau$
maximizes the quantity $q^{\sum_{i=1}^\len (N + \min\{n + \tau_\len, N\} - \tau'_i)\tau'_i}$
and the fact that the number of shapes
$\tau'$ with $\tau'  \preceq \tau$ is upper-bounded by
$\binom{\tau_\len+\len}{\len}$.  Therefore, we have
\ifCLASSOPTIONonecolumn
\[
\log_q \sum_{\tau' \preceq \tau} |\calT_{\tau'}(R^{N \times \min\{n + \tau_\len, N\}})| \!\le\!
\sum_{i=1}^\len (N + \min\{n + \tau_\len, N\} - \tau_i)\tau_i + \len \log_q 4 + \log_q \mbox{$\binom{\tau_\len+\len}{\len}$}.
\]
\else
\begin{multline*}
\log_q \sum_{\tau' \preceq \tau} |\calT_{\tau'}(R^{N \times \min\{n + \tau_\len, N\}})| \!\le\!  \len \log_q 4 + \log_q \mbox{$\binom{\tau_\len+\len}{\len}$}\\
+\sum_{i=1}^\len (N + \min\{n + \tau_\len, N\} - \tau_i)\tau_i.
\end{multline*}
\fi
Combining all the above results, we have obtained the upper
bound.  In particular, when $\mu \succeq 2N$ and $\tau = t$,
we have $\xi_i =
\min\{ n+t, N \}$ for all $i$. Substituting this into the upper bound
completes the proof.
\end{proof}

We next study the asymptotic behavior of ${C}_{\text{AMMC}}$.
\begin{thm}\label{thm:asym-cap-ammc}
When $\mu \succeq 2N$ and $\tau = t$, the asymptotic capacity
$\bar{C}_{\text{MAMC}}$ is upper-bounded by
\begin{equation}\label{eq:asym-cap-ammc}
    \bar{C}_{\text{MAMC}}\le
    \begin{cases}
    \frac{\sum_{i = 1}^\len \bar{n} (\bar{\mu}_i - \bar{n} - \bar{t})}{\bar{n} |\bar{\mu}|} &\mbox{if } n+t \le N\\
    \frac{\sum_{i = 1}^\len (\bar{N} - \bar{t}) (\bar{\mu}_i - \bar{N})}{\bar{n} |\bar{\mu}|} &\mbox{if } n+t > N.
    \end{cases}
\end{equation}
\end{thm}
\begin{proof}
This follows directly from Corollary~\ref{cor:cap-ammc}.
\end{proof}

\subsection{A Coding Scheme}

We again focus on the special case when $\mu \succeq 2N$ and $\tau = t$.
We describe a coding scheme that achieves the asymptotic bound in
Theorem~\ref{thm:asym-cap-ammc}.

\subsubsection{Encoding}
The encoding is a combination of the encoding
strategies for the MMC and the AMC. We first consider
the case when $n+t > N$.
Set $v \ge
t$. We construct the input matrix $X$ as
\[
        {X} = \mat{0 & 0 \\
        0 & \bar{X}},
\]
where the size of $\bar{X}$ is $(N - v) \times (m - v)$, and
the sizes of other zero matrices are readily available.  Here,
$\bar{X}$ is chosen from the set of principal row canonical
forms for $\calT_{\kappa}(R^{(N - v) \times (\mu - v)})$ by
using the construction in Section~\ref{sec:construction},
where $\kappa_i = \min\{ N - v, \lfloor (\mu_i - v) /2 \rfloor
\}$ for all $i$.  The encoding is illustrated in
Fig.~\ref{fig:ammc-scheme}.  Clearly, the encoding rate of the
scheme is $R_\text{MAMC} = \sum_{i = 1}^\len \kappa_{i}(\mu_i
- v - \kappa_i)$.  In particular, when $\mu \succeq 2N$, we
  have $\lfloor (\mu_i - v)/2 \rfloor \ge n- v$ for all $i$.
Thus, $\kappa_i = N - v$ for all~$i$, and  the encoding rate
is $R_\text{MAMC} = \sum_{i = 1}^\len (N - v) (\mu_i - N)$.

\begin{figure}[thb]
  \centering
\ifCLASSOPTIONonecolumn
  \setlength{\tabcolsep}{1pt}
  \begin{tabular}{ccc}
     \includegraphics{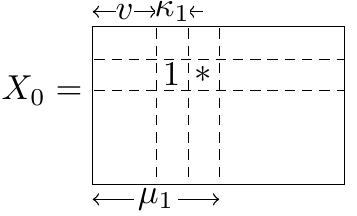} &
     \includegraphics{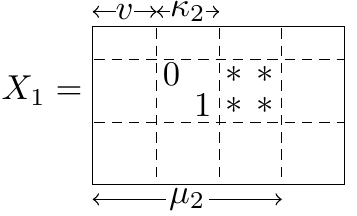} &
     \includegraphics{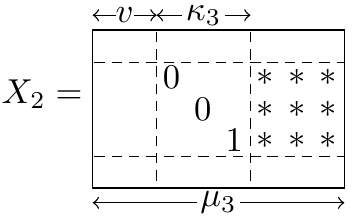}
   \end{tabular}
\else
  \scalebox{0.75}{\setlength{\tabcolsep}{1pt}\begin{tabular}{ccc}
     \includegraphics{fig7a} &
     \includegraphics{fig7b} &
     \includegraphics{fig7c}
   \end{tabular}}
\fi
  \caption{Illustration of the MAMC encoding scheme for $\len = 3$, $N=6$, $n = 5$, $v=2$, $\mu=(4,6,8)$, so that $\kappa = (1,2,3)$.}
  \label{fig:ammc-scheme}
\end{figure}

We then consider the case when $n + t \le N$. Similarly, set $v \ge t$.
We construct the input matrix $X$ as
\[
X = \mat{0 & \bar{X}},
\]
where the size of $\bar{X}$ is $n \times (m - v)$. Again, $\bar{X}$
is chosen from the set of principal row canonical forms for
$\calT_\kappa(R^{n \times (m - v)})$, where $\kappa_i = \min\{
n, \lfloor (\mu_i - v) \rfloor \}$ for all $i$.
Clearly, the encoding rate is $R_\text{MAMC} = \sum_{i= 1}^\len
\kappa_i(\mu_i - v - \kappa_i)$. In particular, when $\mu
\succeq 2N$, we have $\kappa_i = n$ for all $i$, and the encoding
rate $R_\text{MAMC} = \sum_{i= 1}^\len n(\mu_i -n - v)$.

\subsubsection{Decoding}

The decoder receives $Y = P\left(\mat{0\\X} + W \right)$ and attempts to recover
$\bar{X}$ from the row canonical form of $Y$.  We decompose
the noise matrix $W$ as
\[
W = BE = \mat{B_1 \\ B_2} \mat{E_1 & E_2},
\]
as we did in Section~\ref{sec:amc}. Clearly, we have
\[
\mat{0\\X} + W = \mat{B_1 E_1 & B_1 E_2 \\ B_2 E_1 & \bar{X} + B_2 E_2}.
\]
Following \cite{SKK10}, we define error trapping to be
successful if $\shape B_1 E_1 = t$. Assume that this is the
case. From Section~\ref{sec:amc}, there exists some matrix $T
\in \GL_N(R)$ such that
\[
    T\left(\mat{0\\X} + W \right) = \mat{B_1 E_1 & B_1 E_2 \\ 0 & \bar{X}} = \mat{B_1 & 0 \\ 0 & I}
    \mat{E_1 & E_2 \\ 0 & \bar{X}}.
\]
Note that
\[
    \RCF \left( \mat{E_1 & E_2 \\ 0 & \bar{X}} \right)
    = \mat{\tilde{Z}_1 & \tilde{Z}_2 \\ 0 & \bar{X}}
\]
for some $\tilde{Z}_1 \in R^{t \times v}$ in row canonical
form and some $\tilde{Z}_2 \in R^{t \times (m - v)}$. It
follows that
\begin{align*}
    \RCF\left(\mat{0\\X} + W \right) &= \RCF\left( \mat{B_1 & 0 \\ 0 & I}
    \mat{E_1 & E_2 \\ 0 & \bar{X}} \right) \\
    &= \mat{\tilde{Z}_1 & \tilde{Z}_2 \\ 0 & \bar{X} \\ 0 & 0}.
\end{align*}

Since $P$ is invertible, $\RCF(Y) = \RCF\left(\mbox{$\mat{0\\X}$} + W \right)$, from
which $\bar{X}$ can be readily obtained.  Hence, decoding
amounts to computing the row canonical form, whose complexity
is $\calO(nm \min\{ n, m \})$ basic operations over $R$.

The decoding can be summarized as follows.  First,  the
decoder computes $\RCF(Y)$.  Second, the decoder checks the
condition $\shape B_1 E_1 = t$.  If the condition does not
hold, the decoder declares a failure.  Otherwise, the decoder
outputs $\bar{X}$ from $\RCF(Y)$.

Let $n' = \min\{n+v, N\}$.
Let $\hat{Y}$ denote the left-most $n'$ columns of $\RCF(Y)$,
i.e., $\hat{Y} = \RCF(Y)[1{:} N, 1 {:} n']$. We note
that $\shape B_1 E_1 = t$ if and only if $\shape \hat{Y} = t +
\kappa$.  Hence, the error probability of the scheme is zero,
and the failure probability $P_f$ of the scheme is bounded by
$P_f < \frac{2t}{q^{1 + v - t}}$ (as shown in
Section~\ref{sec:amc}).

Finally, if
we set $v$ such that $v - t \to \infty$ and $\frac{v - t}{m}
\to 0$, as $m \to \infty$, we have $P_f \to 0$, and
$\bar{R}_\text{MAMC} = \frac{R_\text{MAMC}}{n | \mu |}$
approaches the upper bound of the asymptotic capacity
in Theorem~\ref{thm:asym-cap-ammc}.

\begin{thm}
When $\tau = t$ and $\mu \succeq 2N$, the coding scheme
described above can achieve the upper bound
(\ref{eq:asym-cap-ammc}).
\end{thm}

\section{Extensions}\label{sec:extension}

Previously, we assume that the transfer matrix
$A \in R^{N \times n}$ is uniform over all full-rank matrices,
and the noise matrix $Z \in R^{N \times m}$ is uniform over
all rank-$t$ matrices. In this section, we discuss possible extensions of
our previous channel models.

\subsection{Non-Uniform Transfer Matrices}

We note that the uniformness assumption on $A$ leads to a
``worst-case'' scenario. To see this, let us
consider a model identical to the MAMC
except for the fact that the transfer matrix $A$
is chosen according to an arbitrary probability
distribution on all full-rank matrices in $R^{N \times n}$.
It should be clear that the capacity of this channel
cannot be smaller than that of the MAMC.
This is because
our coding scheme does not rely on any particular
distribution of $A$ (as long as $A$ is full-column-rank and $Z$
is uniform over all rank-$t$ matrices),
and therefore still works for non-uniform distributions.
Hence, we have the following lower bound on the asymptotic
capacity $\bar{C}$:
\begin{equation}\label{eq:lower-bound}
    \bar{C} \ge
    \begin{cases}
    \frac{\sum_{i = 1}^\len \bar{n} (\bar{\mu}_i - \bar{n} - \bar{t})}{\bar{n} |\bar{\mu}|} &\mbox{if } n+t \le N\\
    \frac{\sum_{i = 1}^\len (\bar{N} - \bar{t}) (\bar{\mu}_i - \bar{N})}{\bar{n} |\bar{\mu}|} &\mbox{if } n+t > N.
    \end{cases}
\end{equation}

On the other hand, the capacity of the channel $Y = AX + Z$
can be upper-bounded by assuming that the transfer matrix $A$
is known at the receiver. One can show that the
asymptotic capacity is upper-bounded by
\begin{equation}\label{eq:upper-bound}
    \bar{C}\le
    \begin{cases}
    \frac{\sum_{i = 1}^\len \bar{n} (\bar{\mu}_i - \bar{t})}{\bar{n} |\bar{\mu}|} &\mbox{if } n+t \le N\\
    \frac{\sum_{i = 1}^\len (\bar{N} - \bar{t}) (\bar{\mu}_i - \bar{t})}{\bar{n} |\bar{\mu}|} &\mbox{if } n+t > N.
    \end{cases}
\end{equation}
Note that when $\mu_1$ is much larger than $N$, the difference
between the lower bound \eqref{eq:lower-bound} and the upper bound
\eqref{eq:upper-bound} is small. In this case,
our coding scheme is close to the capacity.

\subsection{Noise Matrix with Variable Rank}

We consider a more general case where the number of error
packets is allowed to vary, while still bounded by $t$.
More precisely, we assume that $Z$ is chosen uniform at
random from rank-$T$ matrices, where $T \in \{0, \ldots, t\}$
is a random variable with an arbitrary probability distribution
$\Pr[T = k] = p_k$.
Note that
\begin{align*}
H(Z) &= H(Z, T) = H(T) + H(Z|T)\\
&= H(T) + \sum_k p_k H(Z | T = k) \\
&= H(T) + \sum_k p_k \log_q |\calT_k(R^{N \times \mu})| \\
&\le H(T) + \log_q |\calT_t(R^{N \times \mu})|.
\end{align*}
Hence, the capacity may be reduced by at most
$H(T) \le \log_q (t+1)$ compared to the MAMC.
This loss is asymptotically
negligible for large $n$ and $N$.

The coding scheme remains the same. The only difference is
that now decoding errors may occur, because the condition
$\shape B_1 E_1 = t$ becomes $\shape B_1 E_1 = T$, which is,
in general, impossible to check.
Yet, the analysis of decoding is still
applicable, and the error probability is bounded by
$P_e < \frac{2t}{q^{1 + v - t}}$, which goes to $0$ as
$v - t \to \infty$.

\subsection{Non-uniform Noise Matrices}

We note that the uniformness assumption on $Z$ again gives
a ``worst-case'' scenario. To see this,
consider a model identical to the MAMC except for the fact
that the noise matrix $Z$ is chosen according to some non-uniform
probability distribution on $\calT_t(R^{N \times m})$.
It should be clear that the capacity can only increase, since
the entropy $H(Z)$ always decreases.

To apply our coding scheme in this more general
case, we need some transformation.
At the transmitter side,
let $X = X' Q$, where
$Q \in R^{m \times m}$ is chosen uniformly at random
(and independent of any other variables)
from the set of matrices of the form
\[
Q = \mat{Q'_{\mu_1 \times \mu_1} & 0 \\ 0 & I_{m - \mu_1}}.
\]
Here, $Q'$ is an invertible matrix (of size $\mu_1 \times \mu_1$)
and $I$ is an identify matrix (of size $(m-\mu_1) \times (m-\mu_1)$).
Clearly, $Q$ is invertible by construction.
At the receiver side, let $Y' = P Y Q^{-1}$, where $P \in R^{N \times N}$
is chosen uniformly at random (and independent of any other variables)
from all invertible matrices. Then
\begin{align*}
Y' &= P Y Q^{-1} = P (A X' Q + Z) Q^{-1} \\
&= (PA) X' + P Z Q^{-1}.
\end{align*}
After this transformation, our coding scheme can be applied directly.
Moreover, our error analysis still holds, and the failure probability
is again bounded by $P_f < \frac{2t}{q^{1 + v - t}}$.

\section{Conclusions}\label{sec:conclusion}

In this work, we have studied the matrix channel $Y = AX + BE$
where the packets are from the ambient space $\Omega$ of form
\eqref{eq:ambient1}.
Under the assumption that $A$ is
uniform over all full-rank matrices and $BE$ is uniform over
all rank-$t$ matrices, we have derived tight capacity
results and provided polynomial-complexity capacity-achieving
coding schemes, which naturally extend the work of
\cite{SKK10} from finite fields to certain finite rings. Our
extension is based on several new enumeration results and construction methods, for
matrices over finite chain rings, which may be of independent
interest.

We believe that there is still much work to be done in this
area.  One direction would be to further relax the assumptions on $A$
and $BE$.  Following this direction, we have explored a
particular case when $A$ can be any matrix and $BE = 0$ in
\cite{NFSU-submitted}.  Another direction would be to find
other applications of the algebraic tools developed in this paper,
especially the row canonical form.

\appendix

\subsection{Rings and Ideals}\label{sec:rings}

Let $R$ be a ring.  We will let $R^*$ denote the
nonzero elements of $R$, i.e., $R^* = R \setminus \{ 0 \}$.
An element $a$ in $R$ is called a \emph{unit} if $ab = 1$ for
some $b \in R$.  We will let $U(R)$ denote the units in $R$.
Two elements $a, b \in R$ are said to be \emph{associates} if
$a = ub$ for some $u \in U(R)$.  Associatedness is an
equivalence relation on $R$.

Suppose $a, b \in R$. The element $a$ \emph{divides} $b$,
written $a \mid b$, if $ac = b$ for some $c \in R$. Let $d \in
R^*$ be a nonzero element in $R$.  Two elements $a, b$ are
said to be \emph{congruent modulo $d$} if $d$ divides $a - b$.
Congruence modulo $d$ is an equivalence relation on $R$.  A
set containing exactly one element from each equivalence class
is called a \emph{complete set of residues} with respect to
$d$, and is denoted by $\calR(R, d)$.  Note that the
difference $a-b$ between distinct elements $a,b \in
\calR(R,d)$, $a\neq b$, can never be a multiple of $d$.

An element $a$ of $R^*$ is a called a \emph{zero-divisor} if
$ab = 0$ for some $b \in R^*$.  If $R$ contains no
zero-divisors, then $R$ is an \emph{integral domain}.  If $R$
is finite and an integral domain, then $R$ is, in fact, a
finite field.  This latter case is not of central interest in
this paper;  almost all of the rings considered here will have
zero divisors.

\begin{exam}
Let $R = \ZZ_8 \triangleq \{ 0, \ldots, 7 \}$, under integer
addition and multiplication modulo 8.  Then $U(\ZZ_8) = \{1,
3, 5, 7\}$.  There are four equivalent classes induced by
congruence modulo $4$, namely, $\{ 0, 4 \}$, $\{ 1, 5 \}$, $\{
2, 6 \}$, and $\{3, 7\}$. An example of a complete set of
residues with respect to the element $4$ in $\calR(\ZZ_8)$ is
$\calR(\ZZ_8, 4)= \{ 0, 1, 2, 3 \}$.  The zero-divisors of
$\ZZ_8$ form the set $\{ 2, 4, 6 \}$.
\end{exam}

A nonempty subset $I$ of $R$ that is closed under subtraction,
i.e., $a,b \in I$ implies $a-b\in I$, and closed under
inside-outside multiplication, i.e., $a\in I$ and $r \in R$
implies $ar \in I$, is called an \emph{ideal} of $R$.  If $A =
\{ a_1, \ldots, a_m \}$ is a finite nonempty subset of $R$, we
will use $\langle a_1, \ldots, a_m \rangle$ to denote the
ideal generated by $A$, i.e.,
\[
 \langle a_1, \ldots, a_m  \rangle =
 \{ a_1 c_1 + \cdots + a_m c_m \colon
   c_1, \ldots, c_m \in R \}.
\]
An ideal $I$ of $R$ is said to be \emph{principal} if $I$ is
generated by a single element in $I$, i.e., $I = \langle a
\rangle$ for some $a \in I$.  A ring $R$ is called a
\emph{principal ideal ring} (PIR) if every ideal $I$ of $R$ is
principal.  If $R$ is a PIR and also an integral domain, then
$R$ is called a \emph{principal ideal domain} (PID).

An ideal $N$ is said to be \emph{maximal} if $N \neq R$ and
the only ideals containing $N$ are $N$ and $R$ (in other
words, $N$ is ``maximal'' with respect to set inclusion among
all proper ideals). If $N$ is a maximal ideal, then the
quotient $R/N$ is a field, called a \emph{residue field}.  A
ring with a unique maximal ideal is said to be \emph{local}.

\begin{exam}
The ideals of $\ZZ_8$ are $\{ 0 \} = \langle 0 \rangle$, $\{
0, 4 \} = \langle 4 \rangle$, $\{0, 2, 4, 6 \} = \langle 2
\rangle$, and $R = \langle 1 \rangle$.  Thus, $\ZZ_8$ is a
PIR, and has a unique maximal ideal $\langle 2 \rangle$.  The
residue field $\ZZ_8 / \langle 2 \rangle$ is isomorphic to the
finite field $\mathbb{F}_2$ of two elements.
\end{exam}

\subsection{Proofs for Section~\ref{sec:canonical}}\label{sec:canonical-proofs}

\subsubsection{Proof of Proposition~\ref{prop:rcfprops}}
We prove the claims one by one.
\begin{enumerate}
\item The presence of a pivot $p$ in a column rules out the
possibility of another pivot in the same column and below $p$,
since all entries in the same column below $p$ must be zero
and hence cannot be pivots.

\item Deleting a row of $A$ does not influence the value or the
position of the pivots in the other rows; thus it easy to
verify that the modified matrix satisfies the four conditions
required for a matrix to be in row canonical form.

\item By definition $p_k$ has degree smaller than or equal to
that of any element in its row.  If $A$ contained an element
in a row below row $k$ of degree smaller than $d_k$, then the
pivot of that row would have degree smaller than $d_k$,
contradicting the property that pivots of smaller degree must
occur above pivots of larger degree.

\item By definition $p_k$ is the earliest element having minimum
degree in row $k$, so every element in row $k$ occurring
earlier than $p_k$ has degree strictly larger than $d_k$.  We
know from 3) that $A$ contains no element in a row below $k$
of degree smaller than $d_k$.  If such a row contains an
element of degree equal to $d_k$, then the pivot of that row
must occur later than $p_k$, which implies that every element
occurring in that row occurring in column $c_k$ or earlier has
degree strictly larger than $d_k$.

\item Consider $w_j$.  From 3) we know that $p_1$ divides every
element of $A$; in particular, $p_1$ divides every element of
column $j$ of $A$.  Since $w_j$ is a linear combination of
these elements, it must be that $p_1$ divides $w_j$.

\item If $j < c_1$, we know from 4) that every element in column
$j$ of $A$ has degree strictly greater than $d_1$ and so does
every linear combination of these elements, in particular
$w_j$.

\end{enumerate}

\subsubsection{Proof of Proposition~\ref{prop:uniqueness}}

If $A$ is the zero matrix, then its row canonical form must
also be the zero matrix, which is therefore unique.  Thus let
us assume that $A$ is nonzero.

We will proceed by induction on $n$. For $n = 1$, the proof is
obvious.  Thus suppose that $n > 1$, and let $B$ and $C$ be
two row canonical forms of $A$.  Clearly, $\row B = \row C$,
and each row of $B$ and $C$ are elements of $\row A$.  Let
$B[1, j_1]$ and $C[1, j_2]$ be the pivots in the first row of
$B$ and $C$, respectively.  From
Proposition~\ref{prop:rcfprops}--5 we have that $B[1,j_1] \mid
C[1,j_2]$ and $C[1,j_2] \mid B[1,j_1]$; thus $B[1,j_1]$ and
$C[1,j_2]$ are associates.  However, since pivot elements must
take the form $\pi^l$ for some $l$, we conclude that
$B[1,j_1]=C[1,j_2]$.  Suppose $j_1 < j_2$. By
Proposition~\ref{prop:rcfprops}--6 we have $\deg(B[1,j_1]) >
\deg(C[1,j_2])$, contradicting the fact that
$B[1,j_1]=C[1,j_2]$.  A similar contradiction arises if $j_1 >
j_2$.  We conclude that $j_1=j_2$, i.e., both $B$ and $C$ must
have exactly the same pivot element in exactly the same
position in their first row.

Now let $j_1 = j_2 = j$.  Consider the submodule of $\row A$
in which every element has zero in its $j$th component.  Every
element $a$ of this submodule is a linear combination
\[
a = \sum_{i=1}^n b_i B[i,{:}];
\]
for some choice of coefficients $b_1, \ldots, b_n$.  However,
since $a_j = 0$, and $B[i,j] = 0$ for $i > 2$, we must have
$b_1 B[1,j] = 0$.  Since $B[1,j]$ is the pivot element of the
first row of $B$, it divides every element of that row; thus
if $b_1 B[1,j]=0$, then $b_1 B[1,{:}]=0$, i.e., the first row
can only contribute 0 to $a$.   This means that the given
submodule is equal to $\row B[2{:}n,1{:}m]$.  Similarly, the given
submodule is also equal to $\row C[2{:}n,1{:}m]$.  By
Proposition~\ref{prop:rcfprops}--2, both $B[2{:}n, 1{:}m]$ and
$C[2{:}n, 1{:}m]$ are in row canonical form.  Thus by induction,
we have $B[2{:}n, 1{:}m] = C[2{:}n, 1{:}m]$.  This implies that
$B$ and $C$ can differ in their first row only.

Let us assume that $B[1,{:}] \ne C[1,{:}]$, i.e., that the first
rows of $B$ and $C$ are not equal, so that $\Delta =
(\delta_1, \ldots, \delta_m)= B[1,{:}] - C[1,{:}]$ is nonzero.
Since $\Delta$ is an element of $\row A$ with zero in its
$j$th component, we have $\Delta \in \row B[2{:}n, 1{:}m]$, from
which it follows that
\[
   \Delta = \sum_{i = 2}^{n} c_i B[i,{:}],
\]
for some $c_2, \ldots, c_{n} \in R$.  If $B[2{:}n, 1{:}m]$ is
the zero matrix, then $\Delta = 0$, which is a contradiction.
Otherwise, let $B[2, j_3]$ be the pivot of $B[2,{:}]$.  Note,
on the one hand, that $B[i,j_3]=0$ for all $i > 2$; thus
$\delta_{j_3} =  c_2 B[2, j_3]$, i.e., $\delta_{j_3}$ must be
a multiple of $B[2,j_3]$.  On the other hand, because $B[2,
j_3]$ and $C[2, j_3]$ are (identical) pivots, $B[1, j_3]$,
$C[1, j_3] \in \calR(R, B[2, j_3])$.  If $B[1,j_3]$ and
$C[1,j_3]$ are distinct, their difference, $\delta_{j_3}$,
cannot be a multiple of $B[2,j_3]$.  We conclude that
$\delta_{j_3} = 0$, i.e., $B[1,j_3]$ and $C[1,j_3]$ are not
distinct.  Since $B[2, j_3]$ is the pivot of $B[2,{:}]$ it
divides every element of $B[2,{:}]$; thus if $c_2 B[2, j_3] =
0$, then $c_2 B[2,{:}] = 0$.  Continuing this argument, we have
$c_i B[i,{:}] = 0$ for all $i \geq 2$.  Therefore, we have
$\Delta=0$, which is a contradiction.  This establishes
uniqueness.

\subsection{Proofs for Section~\ref{sec:matrix-constraint}}\label{sec:appendix-proofs}

\subsubsection{Proof of Lemma~\ref{lem:correspondence}}
Let $\calS$
denote the set of row canonical forms in $\calT_{\kappa}(R^{n
\times \mu})$, and let $\calG$ denote the set of submodules of
$R^{\mu}$ with shape $\kappa$.  Let $\phi: \calS \to \calG$ be
the map that takes a matrix $B \in \calS$ to its row module
$\row B$.  We will show that $\phi$ is a one-to-one
correspondence.

If $\phi(B_1) = \phi(B_2)$ then $B_1$ and $B_2$ are
left-equivalent, and so $B_2$ is a row canonical form of $B_1$
and vice-versa.  By the uniqueness of the row canonical form,
we have $B_1 = B_2$; thus $\phi$ is injective.

Now let $M$ be a submodule of $R^{\mu}$ with $\shape M =
\kappa$, and construct a matrix $A$ such that every element in
$M$ is a row of $A$.  Clearly, $\row A = M$ and $\shape A =
\kappa$. Since $\kappa \preceq n$, $\RCF(A)$ has at most
$n$ nonzero rows.
Let $B$ be the submatrix of $\RCF(A)$ consisting of the top
$n$ rows.  Then we have $\shape B = \shape A = \kappa$.
Hence, $B \in \calT_{\kappa}(R^{n \times \mu})$, and the map
$\phi$ is surjective.

\subsubsection{Proof of Proposition~\ref{prop:principal}}
We will show that (i) every $X$ constructed as above is a
principal row canonical form, and (ii) every principal row
canonical form has a $\pi$-adic decomposition following the
above conditions.

We begin with Claim~(i).  First, we track the diagonal entries
in $X$. Clearly, by construction, the first $\kappa_1$
diagonal entries in $X$ are $1$; they are contributed by
$X_0$.  The next $\kappa_2 - \kappa_1$ diagonal entries in $X$
are $\pi$; they are contributed by $X_1$.  Continuing this
argument, we conclude that the diagonal entries in $X$ are
indeed of the form (\ref{eq:principal-form}).

Second, we show that $X$ satisfies all the four conditions for
row canonical forms.
\begin{enumerate}
\item By construction, the first $\kappa_\len$ rows of $X$ are
the only nonzero rows. Hence, $X$ satisfies Condition~1.
\item It suffices to show that the nonzero diagonal entries
are precisely the pivots in $X$.  Suppose that the $i$th
diagonal entry $X[i, i] = \pi^{l}$.  Then by construction,
$\pi^l$ is contributed by $X_{l}$ and $\kappa_l < i \le
\kappa_{l+1}$.  Note that for each auxiliary matrix
$X_{l\rq{}}$, only the first $\kappa_{l\rq{}+1}$ rows are
nonzero.  Thus, the $i$th row in $X_{l\rq{}}$ is zero for all
$l\rq{} = 0, \ldots, l-1$. In particular, $X_{l\rq{}}[i, j] =
0$, for all $l\rq{} = 0, \ldots, l-1$ and for all $j > i$.
Therefore, we have, for all $j > i$,
\begin{align*}
    X[i, j] &= \sum_{l\rq{} = 0}^{\len - 1} \pi^{l\rq{}} X_{l\rq{}}[i, j] \\
      &= \sum_{l\rq{} = l}^{\len - 1} \pi^{l\rq{}} X_{l\rq{}}[i, j] \\
      &= \pi^l \sum_{l\rq{} = l}^{\len - 1} \pi^{l\rq{}-l} X_{l\rq{}}[i, j].
\end{align*}
That is, every $X[i, j]$ is a multiple of $\pi^l$ whenever
$j > i$.  On the other hand, by construction, $X[i, j] = 0$
whenever $j < i$.  It follows that $X[i, i]$ is indeed the
pivot of row $i$.  Hence, $X$ satisfies Condition~2.

\item Since the nonzero diagonal entries are the pivots, $X$
satisfies Condition~3.

\item Suppose that the $i$th pivot $X[i, i] = \pi^{l}$.
Then, we have $\kappa_l < i \le \kappa_{l+1}$.  Note that for
each auxiliary matrix $X_{l\rq{}}$, all other entries in
column $i$ are zero as long as $l\rq{} \ge l$. Thus, we have,
for all $j \ne i$,
\begin{align*}
    X[j, i] &= \sum_{l\rq{} = 0}^{\len - 1} \pi^{l\rq{}} X_{l\rq{}}[j, i] \\
      &= \sum_{l\rq{} = 0}^{l - 1} \pi^{l\rq{}} X_{l\rq{}}[j, i].
\end{align*}
It follows that $X[j, i] \in \calR(R, \pi^l)$ for all $j \ne
i$.  Hence, $X$ satisfies Condition~4.
\end{enumerate}

We turn now to Claim~(ii).  Let $X$ be a principal row
canonical form in $\calT_\kappa(R^{n \times \mu})$.  Then the
diagonal entries in each $X_i$ must satisfy
\[
X_i[1,1], \ldots, X_i[\kappa_{i+1}, \kappa_{i+1}] =
\underbrace{0, \ldots, 0}_{\kappa_i}, \underbrace{1, \ldots, 1}_{\kappa_{i+1} - \kappa_i}.
\]
Moreover, since $X$ satisfies Condition~4, it follows that
each $X_i$ satisfies the first condition described above.
Since $X$ satisfies Condition~2, it follows that
$X_i[\kappa_{i+1}+1{:} n, 1 {:} m]$ is a zero matrix.
Finally, due to the constraints imposed by $\mu$, $X_i[1{:}
n, \mu_{i+1}+1 {:} m]$ is a zero matrix for all $i$.
Therefore, each $X_i$ satisfies the second and third
conditions.  This completes the proof.

\subsubsection{Proof of Theorem~\ref{thm:MatrixCount}}

We need two technical lemmas.
The first lemma is a natural extension of the
well-known rank decomposition.

\begin{lem}\label{lem:decomp1}
Let $B$ be the row canonical form of $A \in R^{n \times m}$.
Let $\tilde{B}$ be the submatrix of $B$ consisting of only
nonzero rows. Then $A$ can be decomposed as a product $P_1
\tilde{B}$ of some full-column-rank matrix $P_1$ and the
matrix $\tilde{B}$.  Moreover, the number of $P_1$ producing
such a decomposition is $q^{n \sum_{i = 1}^{\len-1}
i(\kappa_{i+1} - \kappa_i)}$, where $\kappa = \shape A$.
\end{lem}
\begin{proof}
Since $B$ is the row canonical form of $A$, $A = PB$ for some
invertible matrix $P \in \GL_n(R)$.  Since $\kappa = \shape A
= \shape B$, $B$ has $\kappa_\len$ nonzero rows, and $\tilde{B}
\in R^{\kappa_\len \times m}$.  Let $P = \mat{P_1 & P_2}$, where
$P_1 \in R^{n \times \kappa_\len}$ and $P_2 \in R^{n \times (n
- \kappa_\len)}$. Then we have
\[
   A = PB = \mat{P_1 & P_2} \mat{\tilde{B} \\ 0} = P_1 \tilde{B}.
\]
Since $P$ is invertible,  $P_1$ is full column rank.

Next, we count the number of such decompositions.  Consider
the matrix equation $X \tilde{B} = P_1 \tilde{B}$, in unknown
$X$.  Clearly, the number of decompositions of $A$ is equal to
the number of solutions to this matrix equation.  Let
$\tilde{B}[i, j_i]$ be the pivot of the $i$th row of
$\tilde{B}$, for all $i = 1, \ldots, \kappa_\len$. Then
$\tilde{B}[i, j_i]$ divides the $i$th row of $\tilde{B}$. It
follows that $\tilde{B} = D B'$, where $D =
\diag\left(\tilde{B}[1, j_1], \ldots, \tilde{B}[\kappa_\len,
j_{\kappa_\len}] \right)$, and the $i$th row of $B'$ is equal to
the $i$th row of $\tilde{B}$ divided by $\tilde{B}[i, j_i]$.
Clearly, $B'[i, j_i] = 1$ for all $i = 1, \ldots, \kappa_\len$.
Since $j_1, \ldots, j_{\kappa_\len}$ are all distinct,
$\shape B' = (\kappa_\len, \ldots, \kappa_\len)$, which implies that
$B'$ is full row rank.
By Lemma~\ref{lem:column-rank}, $(X D
- P_1 D)B' = 0$ if and only if $X D - P_1 D = 0$. Hence, $X
\tilde{B} = P_1 \tilde{B}$ if and only if $X D = P_1 D$.
Thus, it suffices to count the number of solutions to $X D =
P_1 D$. Note that $X D = P_1 D$ is equivalent to the following
system of equations
\begin{equation}\label{eq:entry-eq}
        X[i, k] \tilde{B}[k, j_k] = P_1[i, k] \tilde{B}[k, j_k],
        i=1,\ldots,n,~k=1,\ldots \kappa_\len.
\end{equation}
Suppose that $\tilde{B}[k, j_k] = \pi^{l_k}$ for some $0 \le k
< \len$.  Then it is easy to check that the equation $X[i, k]
\pi^{l_k} = P_1[i, k] \pi^{l_k}$ has exactly $q^{l_k}$
solutions for $X[i, k]$.  It follows that (\ref{eq:entry-eq})
has exactly $q^{n(l_1 + \cdots + l_{\kappa_\len})}$ solutions.
Finally, by using the fact that $\sum_{k = 1}^{\kappa_\len} l_k =
\sum_{i = 1}^{\len-1} i(\kappa_{i+1} - \kappa_i)$, we complete
the proof.
\end{proof}

\begin{lem}\label{lem:fixed-RCF}
The number of matrices in $R^{n \times \mu}$ having a given
row canonical form in $\calT_{\kappa}(R^{n \times \mu})$ is
equal to
\[
|R^{n \times \kappa}| \prod_{i = 0}^{\kappa_\len -
1} (1 - q^{i - n}).
\]
\end{lem}
\begin{proof}
Let $B$ be a row canonical form in $\calT_{\kappa}(R^{n \times
\mu})$.  Let $\tilde{B}$ be the submatrix of $B$ consisting of
only nonzero rows.  Clearly, $\tilde{B} \in R^{\kappa_s \times
\mu}$.  We would like to count the number of matrices in $R^{n
\times \mu}$ having the row canonical form $B$.

By Lemma~\ref{lem:decomp1}, every matrix $A$ with
$\RCF(A) = B$ has $q^{n \sum_{i = 1}^{\len - 1} i(\kappa_{i+1}
- \kappa_i)}$ decompositions of the form $A = C \tilde{B}$ for
some full-column-rank $C \in R^{n \times \kappa_s}$.  Hence,
the number of matrices in $R^{n \times \mu}$ having the row
canonical form $B$ is equal to the number of full-column-rank
matrices of size $n \times \kappa_\len$ divided by $q^{n
\sum_{i = 1}^{\len - 1} i(\kappa_{i+1} - \kappa_i)}$, which
can be simplified to $|R^{n \times \kappa}| \prod_{i =
0}^{\kappa_\len - 1} (1 - q^{i - n})$.
\end{proof}

We can partition all the matrices in $\calT_{\kappa}(R^{n
\times \mu})$ based on their row canonical forms: two matrices
belong to the same class if and only if they have the same row
canonical form.  By Lemma~\ref{lem:correspondence},  the
number of such classes is $\submodule{\mu}{\kappa}$.  By
Lemma~\ref{lem:fixed-RCF}, the number of matrices in each
class is $|R^{n \times \kappa}| \prod_{i = 0}^{\kappa_\len -
1} (1 - q^{i - n})$.  Combining these two results gives us
Theorem~\ref{thm:MatrixCount}.

\section*{Acknowledgment}
The authors would like to thank Michael Kiermaier for useful
discussions on the topic of row canonical forms for matrices
over finite chain rings.


\end{document}